\documentclass[10pt,a4paper]{article}
\usepackage{amsmath,amsthm}
\usepackage{amsfonts}
\usepackage{amssymb}

\usepackage[section]{algorithm}
\usepackage{algorithmic}

\newcommand{\gcdp}{\text{gc}\text{d}_p}
\usepackage{url}
\usepackage[auth-lg]{authblk}

\usepackage{hyperref}

\theoremstyle{plain}
\newtheorem{theorem}{Theorem}
\theoremstyle{definition}
\newtheorem{lemma}[theorem]{Lemma}
\newtheorem{definition}[theorem]{Definition}
\newtheorem{remark}[theorem]{Remark}
\newtheorem{notation}[theorem]{Notation}
\newtheorem{corollary}[theorem]{Corollary}

\usepackage{float}

\begin{document}

\title{The Complexity of Computing all Subfields of an Algebraic Number Field}

\author{Jonas Szutkoski\footnote{This work is part of the doctoral studies of Jonas Szutkoski, who acknowledges the support of CAPES, Grant 8804/14-1 - Brazil.} }
\affil{Universidade Federal do Rio Grande do Sul\\ Porto Alegre, RS 91509-900, Brazil\\ jonas.szutkoski@ufrgs.br}

\author{Mark van Hoeij\footnote{Supported by NSF Grants 1319547 and 1618657.}}
\affil{Florida State University\\ Tallahassee, FL 32306, USA\\ hoeij@math.fsu.edu}

\maketitle

\begin{abstract}
For a finite separable field extension $K/k$, all subfields can be obtained by intersecting so-called \textit{principal subfields} of $K/k$. In this work we present a way to quickly compute these intersections. If the number of subfields is high, then this leads to faster run times and an improved complexity. 
\end{abstract}

\section{Introduction}
Let $k$ be a field and let $K=k(\alpha)$ be a separable field extension of degree $n$. Let $f\in k[x]$ be the minimal polynomial of $\alpha$ over $k$. There are several methods to compute the subfields of a field extension. Many of these methods take advantage of the connection between subfields of a field extension and subgroups of the Galois group $\text{Gal}(f)$, see \cite{Dixon}, \cite{Hulpke} and \cite{Kluners}. Other methods involve resolvents, such as \cite{Lazard}, and symmetric functions, \cite{Casp}. The \texttt{POLRED} algorithm \cite{Cohen} may also find subfields, but it is not guaranteed.

According to \cite{HKN}, there exists a set $\{L_1,\ldots, L_r \}$ of so-called principal subfields, with $r\leq n$, such that every subfield of $K/k$ is the intersection of a subset of $\{L_1,\ldots, L_r\}.$
Thus, computing all subfields can be done in two phases

\begin{description}

\item \textbf{Phase I}: Compute $L_1,\ldots, L_r$.

\item \textbf{Phase II}: Compute all subfields by computing intersections of $L_1,\ldots, L_r$.
\end{description}

Principal subfields can be computed (see \cite{HKN}) by factoring polynomials over $K$ and solving linear equations. If $k=\mathbb{Q}$, one can also use a $p$-adic factorization and LLL (see Subsection~\ref{ssec:compare} for a comparison between these two approaches).

Phase I usually dominates the CPU time. However, in the theoretical complexity, the reverse is true: for $k=\mathbb{Q}$, Phase~I is polynomial time but Phase~II depends on the number of subfields, which is not polynomially bounded.

The goal of this paper is to speed up Phase~II. This improves the complexity (Theorem~\ref{theo:principal}) as well as practical performance, although the improvement is only significant when the number of subfields is large (see Section~\ref{sec:6}).

This paper is organized as follows. Section~\ref{sec:2} associates to each subfield $L$ a partition $P_L$ of $\{1,\ldots, r\}$. Subfields can be intersected
efficiently when represented by partitions (Section~\ref{ssec:part}).
Section~\ref{sec:3} shows how to compute the partition for principal subfields. Section~\ref{sec:4} presents a general algorithm for computing all subfields of a finite separable field extension. We give another algorithm for the case $k=\mathbb{Q}$ in Section~\ref{sec:5} and compare algorithms in Section~\ref{sec:6}.

To analyse the complexity we will often use the \textit{soft}-$\mathcal{O}$ notation $\tilde{\mathcal{O}}$, which ignores logarithmic factors. We also make the following assumptions: 
\begin{itemize}
\item[1)] Given polynomials $f$, $g$ of degrees at most $n$, we can compute $fg$ with at most $\tilde{\mathcal{O}}(n)$ field operations (\cite{MCA}, Theorem~8.23).
\item[2)] If $K=\mathbb{Q}[\alpha]$ is an algebraic number field of degree $n$, field operations in $K$ can be computed with $\tilde{\mathcal{O}}(n)$ field operations in $\mathbb{Q}$ (\cite{MCA}, Corollary~11.8).
\item[3)] Let $2 \leq \omega \leq 3$ denote a \textit{feasible matrix multiplication} exponent ($\omega \leq 2.3728$, see \cite{LeGall}).
For a linear system over $K$ with $m$ equations and $r\leq m$ unknowns, we assume one can compute a basis of solutions with $\mathcal{O}(mr^{\omega-1})$ operations in $K$ (\cite{BP94}, Chapter~2).
\end{itemize}

\subsection{Notations}
Throughout this paper $K/k$ will be a finite separable field extension with primitive element $\alpha$ and $f\in k[x]$ will denote the minimal polynomial of $\alpha$ over $k$. The degree of the extension $K/k$ will be denoted by $n$.

Let $\hat{K}$ be an extension of $K$ and let $f = \hat{f_1} \cdots \hat{f}_{\hat{r}} \in \hat{K}[x]$ be the factorization of $f$ in $\hat{K}[x]$. Let $g\in \hat{K}[x]$ with $g\mid f$. Since $K/k$ is separable, $g$ is separable as well. The following set is a subfield of $K$ (see \cite[Section~2]{HKN})
\begin{equation}\label{eq:Lg}
L_g := \left\{ h(\alpha) \ : \ h(x)\in k[x]_{<n} {\rm \ \,with \ \,} h(x)\equiv h(\alpha)\bmod g\right\}
\end{equation}
where $k[x]_{<n}$ is the set of polynomials over $k$ with degree at most $n-1$.
% and $n$ is the degree of the extension $K/k$.

\begin{remark} \label{remark1} 
If $g=g_1 g_2\mid f$ then $L_g = L_{g_1} \cap L_{g_2}$ (Chinese Remainder Theorem).
\end{remark}

\begin{definition}
If $g$ is irreducible in $\hat{K}[x]$, then $L_g$ is called a \emph{principal subfield}.
The set $\{ L_{\hat{f_1}}, \ldots, L_{\hat{f}_{\hat{r}}} \}$ is independent (see \cite{HKN}) of the choice of $\hat{K}\supseteq K$ and is called the set of \emph{principal subfields} of $K/k$. 
\end{definition}

\begin{theorem}[\cite{HKN}, Theorem~1]\label{theo:principal1}
Let $L$ be a subfield of $K/k$. Then there exists a set $I\subseteq \{1,\ldots,\hat{r} \}$ such that \[ L =\bigcap_{i \in I} L_{\hat{f}_i}. \]
\end{theorem}

\begin{theorem}\label{theo:equiv}
Let $g\in K[x]$ monic be such that $x-\alpha \mid g \mid f$. The following are equivalent
\begin{enumerate}
\item[$(1)$] $g$ is the minimal polynomial of $\alpha$ over $L$, for some subfield $L$ of $K/k$.
\item[$(2)$] $deg(g)\cdot [k(\text{coeffs}(g)) : k] \leq n$.
\item[$(3)$] $deg(g)\cdot [L_g : k] = n$.
\item[$(4)$] The coefficients of $g$ generate $L_g$ over $k$.
\item[$(5)$] $g\in L_g[x]$.
\item[$(6)$] $g=\gcd(f,h(x)-h(\alpha))$, for some $h(x)\in k[x]_{<n}$.
\end{enumerate}
\end{theorem}

\begin{proof}
See Appendix A.
\end{proof}

\begin{definition}
If any of the conditions in Theorem~\ref{theo:equiv} holds, then $g$ is called a \emph{subfield polynomial}. Furthermore, we call $g$ the \emph{subfield polynomial} of the field $L$ in $(1)$, which coincides with $k(\text{coeffs}(g))$ in $(2)$, $L_g$ in $(3)$, $(4)$, $(5)$ and $k(h(\alpha))$ in $(6).$
\end{definition}

The subfield polynomial of $K$ is $x-\alpha$ and the subfield polynomial of $k$ is $f$.

\begin{remark} \label{RemarkKhat} Choosing $\hat{K}$. For $k = \mathbb{Q}$, \cite{HKN} gives two methods for Phase~I.
\begin{enumerate}
\item Take $\hat{K}=K$, i.e, factor $f$ over $K$. Next, \cite{HKN} computes principal subfields with $k$-linear algebra (which we replace with Section~\ref{sec:3} in this paper).
\item Take $\hat{K}=\mathbb{Q}_p$ for suitable $p$. This method avoids factoring $f$ over a number field. Instead, it factors $f$ over the $p$-adic numbers. The principal subfield for each $p$-adic factor is then computed with LLL.
\end{enumerate}
We developed both methods to find out which one works best when adjusted to our new approach for Phase~II.
Based on timings in Magma \cite{magma} (Section~\ref{sec:6}) and the fact that the LLL % cut-off
bound in \cite{HKN} is practically optimal, we expected method~2 to be faster,
until we tried factoring $f$ over $K$ with Belabas' algorithm \cite{belabas} in
%GP/PARI
{PARI/GP} \cite{pari}
	% The LLL cut-off bound for method 2 in \cite{HKN} is essentially optimal, which suggests that it should be be fastest way to compute a principal subfield.
	% But that does not imply it is the fastest way to compute {\em all} principal subfields.
(Method~2 and Belabas' factoring both use LLL, but method~2 does this for each $p$-adic factor separately, introducing a factor $\hat{r}$.)

Without additional results (can LLL-work for $p$-adic factors in method~2 be shared?) method~1 with
% GP/PARI
{PARI/GP}
gives the fastest CPU timings for Phase~I.
However, estimating the complexity (Theorem~\ref{ComplexityMethod2}) is easier for method~2.
%  (method~1 is only faster when combined with \cite{belabas} which does not give a complexity).
\end{remark}

\section{Subfields and Partitions}\label{sec:2}
To illustrate the goal of this Section we start with an example. Let $k = \mathbb{Q}$, $f=x^6-2$, $\alpha$ is a root of $f$, and $K = \mathbb{Q}(\alpha)$. Factor $f$ over $K$:
\begin{equation} \label{eq:ex}
f = f_1 f_2 f_3 f_4 = (x-\alpha) (x + \alpha) (x^2 - \alpha x + \alpha^2) (x^2 + \alpha x + \alpha^2). \end{equation}
With Equation~(\ref{eq:Lg}) we can show that $L_2:=L_{f_2}=\mathbb{Q}(\alpha^2)$.
Over this subfield, the irreducible factors of $f$ are \[g_1=f_1 f_2 = x^2-\alpha^2 \;\;\text{ and }\;\; g_2= f_3 f_4 = x^4 + \alpha^2 x^2 + \alpha^4.\]
This factorization can be encoded with a partition $P_2 := \{\{1,2\},\{3,4\}\}$. Its first part $\{1,2\}$ encodes the subfield polynomial of $L_2$,
and hence encodes $L_2$ by Theorem~\ref{theo:equiv} part~(4).
A subfield $L$ of $K/k$ is usually represented with a basis (as $k$-vector space), or with generator(s).
This Section will represent subfields with partitions instead. The benefit to Phase~II is shown in Theorem~\ref{theo:intersect}.

% Notice that we can determine $L_5$ if we are given the partition $P_5$ and the factorization (\ref{eq:ex}). In general, every subfield of $\mathbb{Q}(\alpha)/\mathbb{Q}$ defines a partition on the factors of $f$. For instance, the partition given by $L_2:=L_{f_2}=\mathbb{Q}(\alpha^2)$ is $P_2=\{\{f_1,f_2\},\{f_3,f_4,f_5\}\}$, and so on. The question we wish to answer is, how can one (quickly) find the partition $P$ defined by the intersection $L_2\cap L_5$ using only $P_2$ and $P_5$? As it turns out (Theorem \ref{theo:intersect}), $P$ is the \textit{join} of $P_2$ and $P_5$.
 
% Since every subfield is the intersection of principal subfields, we would like to quickly compute $L_2\cap L_5$.

\subsection{From a Subfield to a Partition}\label{ssec:sp}
Let $f=f_1\cdots f_r$ be a partial\footnote{If we only use method~1 from Remark~\ref{RemarkKhat}, then parts of Section~\ref{sec:2} can be shortened
by only considering a full factorization of $f$ (irreducible $f_i$).}
factorization of $f$ over $K$ ($f_i$ not necessarily irreducible). In this Section we define a partition $P_L$ of $[r]:=\{1,\ldots, r\}$ for a given subfield $L$ of $K/k$. Recall that a partition $P=\{P^{(1)}, \ldots, P^{(t)}\}$ of $[r]$ satisfies 
\begin{enumerate}
\item $\bigcup P^{(i)}=[r]$.
\item $P^{(i)}\neq \emptyset,\; 1\leq i \leq t$.
\item $P^{(i)} \bigcap P^{(j)} = \emptyset$, for every $i\neq j$.
\end{enumerate}

\begin{notation}\label{not:Pi} The number of parts in a partition $P$ is denoted by $|P|$. We number them in such a way that $1 \in P^{(1)}$ and $ \min([r] - P^{(1)} \cup \cdots \cup P^{(j)}) \in P^{(j+1)}$.
\end{notation}

\begin{definition}
Let $P=\{P^{(1)},\ldots, P^{(t)}\}$ be a partition of $[r].$ We call \emph{$P$-products} (with respect to the factorization $f_1\cdots f_r$ of $f$) the polynomials defined by $\prod_{i\in P^{(j)}} f_i$, for each $j=1,\ldots, t$.
\end{definition}

\begin{definition}\label{def:subfieldpart}
For every subfield $L$ of $K/k$, let $P_L$ be a partition of $[r]$ satisfying
\begin{enumerate}
\item The $P_L$-products are in $L[x]$.\label{item1}
\item $|P_L|$ is maximal satisfying \ref{item1}.
\end{enumerate}
\end{definition}

If $f_1,\ldots, f_r$ are the irreducible factors of $f$ over $K$, then the $P_L$-products are the irreducible factors of $f$ over $L$. After the next lemma we show
that $P_L$ is well-defined. Thus, we may say that $P_L$ is \textit{the} partition of $L$. 

\begin{notation}\label{not:pi}
Let $\{f_1,\ldots,f_r\}^\pi $ denote the set $\{\prod f_i^{e_i} : e_i\in\{0,1\}\}$.
\end{notation}

\begin{lemma}\label{lemma:factor}
Let $L$ be a subfield of $K/k$. Let $P$ be a partition of $[r]$ and let $F_1,\ldots, F_t$ be the $P$-products. If $P$ satisfies Definition~\ref{def:subfieldpart}, then
\begin{equation}\label{eq:partition}
 \{f_1,\ldots, f_r\}^\pi\bigcap L[x] = \{F_1,\ldots, F_t\}^\pi.
 \end{equation}
\end{lemma}
\begin{proof} Since $P$ satisfies Definition~\ref{def:subfieldpart}, it follows that $F_i\in L[x]$ and hence, $\{F_1,\ldots, F_t\}^\pi\subseteq \{f_1,\ldots, f_r\}^\pi\bigcap L[x]$.
Now, let $F\in\{f_1,\ldots,f_r\}^\pi\bigcap L[x].$ Then $\gcd(F,F_i)\in L[x]$, for every $i=1,\ldots, t$. Furthermore, $\gcd(F,F_i)\in \{1,F_i\}$ (otherwise, we could replace $P^{(i)}$ in $P$ by two non-empty sets, which contradicts the maximality of $|P|$). Therefore, $F\in \{F_1,\ldots, F_t\}^\pi$ and Equation~(\ref{eq:partition}) follows. %Conversely, let $P$ be a partition of $\{1,\ldots, r\}$ and assume Equation (\ref{eq:partition}). We need to prove that $P$ satisfies Definition $\ref{def:subfieldpart}$. From Equation (\ref{eq:partition}), it follows that $F_i\in L[x]$. Condition (2) follows from the fact that $f$ is separable and any partition $P$ satisfying Definition \ref{def:subfieldpart} (1) defines $|P|$ multiplicatively independent elements (i.e., pair-wise coprime) of $\{f_1,\ldots, f_r\}^\pi \bigcap L[x]$. By Equation (\ref{eq:partition}), the maximal number of such elements in $\{f_1,\ldots, f_r\}^\pi \bigcap L[x]$ is $|P|$.
\end{proof}

If $P'_L$ also satisfies Definition~\ref{def:subfieldpart}, then clearly $|P_L|=|P'_L|$. Moreover, if $F'_1,\ldots, F'_t$ are the $P'_L$-products, then $\{F'_1,\ldots, F'_t\}^\pi=\{F_1,\ldots, F_t\}^\pi$, by Lemma~\ref{lemma:factor}. In particular, $F'_i\in \{F_1,\ldots, F_t\}^\pi$, $i=1,\ldots, t$. W.l.o.g., suppose that $F_i\mid F'_i$. If $F'_i \neq F_i$, then we can use the same argument as in the proof of Lemma~\ref{lemma:factor} to create a partition $P$ satisfying Definition~\ref{def:subfieldpart} (1) and such that $|P|> |P'_L|$, which contradicts the maximality of $|P'_L|$. Hence, $F_i'=F_i$ and since $f$ is separable, we have $P_L^{(i)}={P'}_{L}^{(i)}, i=1,\ldots,t$. Thus, $P_L=P_L'$.

% That is, $\gcd(F_i,F'_j)\neq 1$, for some $1\leq j\leq t$. We wish to show that $F_i=F'_j$. Since $F_i,F'_j\in L[x]$, then $1\neq \gcd(F_i,F'_j)\in L[x]$. If $F_i\neq F'_j$, then we can use the same argument in the proof of Lemma \ref{lemma:factor} to create a partition $P$ satisfying Definition \ref{def:subfieldpart} with $|P|>|P_L|=|P'_L|$, which contradicts the maximality of $|P_L|$. 

%Since $|P_L|=|P'_L|$ is maximal, we have $\{F_1,\ldots, F_t\}=\{F'_1,\ldots, F'_t\}$.

% To show that $P_L=P'_L$, suppose, w.l.o.g., that $P_L^{(i)}\cap {P'}^{(i)}_L\neq \emptyset$. If $F_i\neq F'_i$, then $1\neq \gcd(F_i,F'_i)\in L[x]$ and again, we could replace $P_L^{(i)}$ by two non-empty sets in $P_L$. The resulting partition would still satisfy Definition \ref{def:subfieldpart} (1), but would contradict the maximality of $|P_L|$. Hence, $F_i=F'_i$ and since $f$ is separable, we have $P_L^{(i)}={P'}_L^{(i)}$.

% Moreover,  
%

\subsection{Benefits of representing Subfields with Partitions}\label{ssec:sf}

Let $f=f_1\cdots f_r$ be a partial factorization of $f$. In Section~\ref{ssec:sp} we showed that every subfield $L$ of $K/k$ defines a partition $P_L$ of $[r]$. However, not every partial factorization of $f$ has the property that different subfields define different partitions. In what follows we give a condition on $f_1,\ldots, f_r$ for which $L\neq L'$ implies $P_L\neq P_{L'}$.

\begin{definition}\label{def:subfact}
Let $f_1,\ldots, f_r \in K[x]$ be a partial factorization of $f$. We say that $f_1,\ldots, f_r$ is a \emph{subfield factorization} of $f$ if $f_1=x-\alpha$ and $\{f_1,\ldots, f_r\}^\pi$ contains the subfield polynomial of every principal subfield of $K/k$.
\end{definition}

The full factorization of $f$ into irreducible factors over $K$ is always a subfield factorization of $f$, but the converse need not be true. For example, if $K/k$ has no nontrivial subfields, then $\{x-\alpha, f/(x-\alpha)\}$ is a subfield factorization of $f$, even if $f/(x-\alpha)$ is reducible.

\begin{remark}
If $f_1,\ldots, f_r$ is a subfield factorization of $f$, then $\{L_{f_1}, \ldots, L_{f_r}\}$ is the set of principal subfields of $K/k$.
\end{remark}

From now on we will assume that $f_1,\ldots, f_r$ is a subfield factorization of $f$. This allows us to prove that $P_L=P_{L'}$ if and only if $L=L'$ (see Theorem~\ref{lemma:subfieldpoly}).

\begin{lemma}\label{lemma:SubPoly}
Let $f_1, \ldots, f_r$ be a subfield factorization of $f$. If $g$ is a subfield polynomial, then $g\in \{f_1,\ldots, f_r\}^\pi$.
\end{lemma}
\begin{proof}
Let $g$ be a subfield polynomial and let $\tilde{g}=\prod_{f_i\mid g} f_i.$ Hence, $\tilde{g}\mid g$. We need to prove that $g\mid \tilde{g}.$
Let $h\in K[x]$ be an irreducible polynomial such that $h\mid g$. Let $\tilde{h}$ be the subfield polynomial of $L_h$. Since $h\mid g$, it follows that $L_g\subseteq L_h$ and hence $h\mid \tilde{h}\mid g$ (See Appendix~\ref{App:A}, Lemma~\ref{lemma:apx}). On the other hand, since $h$ is irreducible, $L_h$ is a principal subfield and by Definition~\ref{def:subfact}, $\tilde{h}\in \{f_1,\ldots, f_r\}^\pi$. Therefore, $\tilde{h}\mid \tilde{g}$ and hence $h\mid \tilde{g}.$
\end{proof}

The next theorem shows that every subfield $L$ is uniquely defined by the partition $P_L$, provided $f_1,\ldots, f_r$ is a subfield factorization.

\begin{theorem}\label{lemma:subfieldpoly}
Let $f_1, \ldots, f_r$ be a subfield factorization of $f$. Let $L$ be a subfield of $K/k$ and let $P_L$ be its partition. Then the {\em first $P_L$-product}, i.e., $\prod_{i\in P_L^{(1)}} f_i$, is the subfield polynomial of $L$. In particular, $P_L=P_{L'}$ if and only if, $L=L'$.
\end{theorem}
\begin{proof}
Let $h$ be the first $P_L$-product and $g$ be the subfield polynomial of $L$. By Definition~\ref{def:subfieldpart} (1), it follows that $h\in L[x]$. Furthermore, since $1\in P_L^{(1)}$, we have $f_1=x-\alpha \mid h$ and hence $g\mid h$, by Theorem~\ref{theo:equiv} (1). By Lemma~\ref{lemma:SubPoly}, we have $g\in \{f_1,\ldots, f_r\}^\pi$. If $\deg(g)< \deg(h)$, then we can replace $P_L^{(1)}$ by two non-empty sets in $P_L$ (one corresponding to the factors of $g$ and the other to $h/g$). This new partition also satisfies Definition~\ref{def:subfieldpart} (1), but contradicts the maximality of $|P_L|$. Thus, if $f_1,\ldots, f_r$ is a subfield factorization, then for every partition $P_L$, the subset $P_L^{(1)}\subseteq \{1,\ldots, r\}$ encodes the subfield polynomial of $L$, which uniquely defines $L$. This means that if $L\neq L'$, then $P_L\neq P_{L'}$. Conversely, if $L=L'$, then $P_L=P_{ L'}$, since $P_L$ is well-defined. \end{proof}

Representing subfields with partitions has many advantages:
\begin{enumerate}

\item Given $P_L$, one can quickly find elements of $L$, for instance, by computing a coefficient of a $P_L$-product, or by computing a $P_L$-product evaluated at $x=c$, for some $c \in k$. Section~\ref{ssec:generators} gives a fast test to see if the obtained elements generate $L$.

\item $P_L^{(1)}$ immediately gives the subfield polynomial in partially factored form.

\item Given $P_{L}$ and $P_{L'}$, it is trivial (see Lemma~\ref{lemma:subset}) to check whether $L \subseteq L'$. Section~\ref{ssec:part} shows that one can quickly compute the partition for $L \bigcap L'$. The degree $[L:k]$ can be read from $P_L^{(1)}$ with Theorem~\ref{theo:equiv} (2).

\item $P_L$ only requires $\mathcal{O}(r\log r)$ bits of storage\footnote{Our implementation represents a partition of $\{1,\ldots,r\}$ with a \textit{partition-vector}  $v = (v_1,\ldots,v_r)$,
where $v_i \in \{1,\ldots,i\}$ is the smallest element of the part that contains $i$. % This way, each partition is represented with $\mathcal{O}(r \log r)$ bits of data.
}. In particular, we have
\end{enumerate}

\begin{lemma}
When a subfield factorization $f_1, \ldots, f_r$ of $f$ is given, one only needs $\mathcal{O}(mr\log r)$ additional bits to represent the complete subfield lattice of $K/k$, where $m$ is the number of subfields of $K/k$.
\end{lemma}

\subsection{Intersecting Subfields represented by Partitions}\label{ssec:part}

This Section determines the partition of the intersection of two subfields $L$ and $L'$ using only their partitions.
That way Phase~II only needs to use objects with very small bit size.

% could make Phase~II much faster than in \cite{HKN} because each partition uses only a small number of bits.
% That could make Phase~II much faster than the linear algebra used in \cite{HKN} because partitions are small objects.

% s are much smaller objects than the $k$-vector space bases used in \cite{HKN}.
% After that, Phase~II can done with partitions, which is much faster than  the $k$-linear algebra used in \cite{HKN}.

% we can replace the $k$-linear algebra that \cite{HKN} uses for Phase~II with much a faster partition operation.

%
% For efficiency purposes, we wish to compute this partition using the partitions of $L$ and $L'$.
% The main advantage in doing so is that no operations over the field $k$ are required.
%
%\subsubsection{The Partition of $L\cap L'$}

\begin{definition}
A partition $P$ is a {\em refinement} of $Q$ (or simply $P$ {\em refines} $Q$) if every $Q^{(i)}$ can be written as a union of some of the $P^{(j)}$.
\end{definition}

\begin{lemma}\label{lemma:subset}
Let $L,L'$ be two subfields of $K/k$ and let $P_{L}$ and $P_{L'}$ be their partitions of $[r]$. Then $L\subseteq L'$ if, and only if, $P_{L'}$ refines $P_L$. 
\end{lemma}
\begin{proof}
If $P_{L'}$ refines $P_L$, then $P_{L'}^{(1)}\subseteq P_L^{(1)}$. This means that the subfield polynomial of $L$ is divisible by the subfield polynomial of $L'$. Remark~\ref{remark1} implies that $L\subseteq L'$. The converse follows from Lemma~\ref{lemma:factor}.
\end{proof}

% \begin{definition}
% Let $P$ be a partition of $[r]$. We say that $P$ is the finest partition satisfying some property $X$ if $P$ satisfies $X$ and, for every partition $Q$ satisfying $X$, $P$ refines $Q$.
% \end{definition}
% 
% --> Still needed?

The refinement relation is a partial ordering on the set $\Pi_r$ of all partitions of $[r]$. Let $P,Q\in \Pi_r$. The finest partition that is refined by both $P$ and $Q$ is called the \textit{join} of $P$ and $Q$, and is denoted by $P\vee Q$. The parts of $P\vee Q$ are the smallest subsets of $[r]$ that are unions of parts of $P$, and of $Q$ (see~\cite{MeetAndJoin} and~\cite{stanley}, Section~3.3).

\begin{theorem}\label{theo:intersect}
Let $L,L'$ be two subfields of $K/k$ and let $P_{L}$ and $P_{L'}$ be their partitions. Then $P_{L\cap L'}=P_L \vee P_{L'}$.
\end{theorem}

\begin{proof} Let $P=P_{L\cap L'}=\{P^{(1)},\ldots,P^{(t)}\}$ satisfy items (1) and (2) of Definition~\ref{def:subfieldpart}. We need to prove that $P$ is the finest partition such that $P_L=\{P_L^{(1)}, \ldots, P_L^{(s)}\}$ and $P_{L'}=\{ P_{L'}^{(1)},\ldots, P_{L'}^{(s')}\}$ refine $P$. 

Since $L\cap L' \subseteq L$ and $L\cap L' \subseteq L'$, Lemma~\ref{lemma:subset} implies that $P_L$ and $P_{L'}$ refine $P=P_{L\cap L'}$. To prove that $P$ is the finest partition with this property, let $Q$ be a partition refined by both $P_L$ and $P_{L'}$. We need to prove that $P$ refines $Q$. 

Pick $Q^{(i)}$ and let $P^{(j)}$ be such that $R:=Q^{(i)}\cap P^{(j)}\neq \emptyset.$ We need to prove that $P^{(j)} \subseteq Q^{(i)}$. Since $P_L$ and $P_{L'}$ refine $P$, there exist subsets $J_1\subseteq \{1,\ldots, s\}$ and $J_2\subseteq \{1,\ldots, s'\}$ such that \[ P^{(j)} = \bigcup_{k\in J_1} P_L^{(k)} = \bigcup_{k\in J_2} P_{L'}^{(k)}. \]
Likewise, there exist $I_1\subseteq \{1,\ldots, s\}$ and $I_2\subseteq \{1,\ldots, s'\}$ such that \[ Q^{(i)} = \bigcup_{k\in I_1} P_L^{(k)} = \bigcup_{k\in I_2} P_{L'}^{(k)}.  \]
Therefore, \begin{equation}\label{eq:1} 
R=Q^{(i)}\cap P^{(j)} = \bigcup_{k\in I_1\cap J_1} P_L^{(k)} = \bigcup_{k\in I_2\cap J_2} P_{L'}^{(k)}
\end{equation}
and 
\begin{equation}\label{eq:2}
P^{(j)} \setminus R = \bigcup_{k\in J_1\setminus I_1} P_L^{(k)} = \bigcup_{k\in J_2\setminus I_2} P_{L'}^{(k)}.
\end{equation}

If $R\neq P^{(j)}$, then we can replace $P^{(j)}$ by the non-empty sets $R$ and $P^{(j)}\setminus R$. Equations~(\ref{eq:1}) and (\ref{eq:2}) imply that the resulting partition is refined by both $P_L$ and $P_{L'}$ and therefore, satisfies item (1) of Definition~\ref{def:subfieldpart} for $L\cap L'$. This contradicts the maximality of $|P_{L\cap L'}|$. Hence, $P^{(j)}\subseteq Q^{(i)}$ and $P$ refines $Q$.
\end{proof}

%%*********** this piece of text instead of subsection 2.2.2 ***********
%%\vspace{0.3cm}
%
% Using this notation, Theorem \ref{theo:intersect} can be restated as
%
%\begin{thmbis}{theo:intersect}\label{theo:bis}
%Let $L,L'$ be two subfields of $K/k$ and let $P_{L}$ and $P_{L'}$ be their corresponding partitions. 
%\end{thmbis}

According to \cite{Freese1} (see also \cite{Freese2}), the join of two partitions can be computed with $\mathcal{O}(r\log r)$ CPU operations.
So after computing the partitions of the principal subfields in Section~\ref{sec:3} below,
we can compute the partition of every subfield by joining partitions, which is much faster than the $k$-vector space intersections used in \cite{HKN}.
% In Section \ref{sec:3} we show how to compute the partitions given by the principal subfields and in
After that, Subsection~\ref{ssec:generators} shows how to find generators for a subfield represented by its partition.

\section{Phase I, computing partition $P_{i}$ of a principal subfield $L_i$}\label{sec:3}

Let $f_1,\ldots, f_r$ be a subfield factorization of $f$. In general, one can compute a subfield factorization of $f$ by factoring $f$ over $K$. For $k=\mathbb{Q}$ we will give an alternative in Section~\ref{sec:5}. In this section we present how one can compute the partition $P_i$ of $\{1,\ldots,r\}$ defined by a principal subfield $L_i$ of $K/k$. To find $P_i$, it suffices to find a basis of the vectors $(e_1,\ldots, e_r)\in \{0,1\} ^r$ for which 
\begin{equation}\label{g}
\prod_{j=1}^r f_j^{e_j}\in L_i[x].
\end{equation}

\begin{remark}
Let $h_j$ be the logarithmic derivative of $f_j$, that is, $h_j=f_j'/f_j \in K(x)$ and let $H(x)=\sum_{j=1}^r e_j h_j$. If $g=\prod_{j=1}^r f_j^{e_j}$, then $g'/g = H$.
\end{remark}

\begin{definition}
Let $f\in k[x]$. Then $f$ is {\em semi-separable} if $\text{char}(k)=0$ or $\text{char}(k)=p$ and $f$ has no roots with multiplicity $\geq p$.
\end{definition}

\begin{lemma}\label{lemma:gprime1}
Let $g\in K[x]$ monic and semi-separable, and let $L$ be a subfield of $K/k$. If $g'/g \in L(x)$, then $g\in L[x]$.
\end{lemma}
\begin{proof}
Consider the groups $(K(x)^*, \cdot )$ and $(K(x),+)$ and let $\phi : K(x)^*\rightarrow K(x)$ be the group homomorphism defined by $\phi(g)=g'/g$. The kernel of $\phi$ is $K^*$ in characteristic 0 and $K(x^p)^*$ in characteristic $p$. So, if we restrict $\phi$ to monic semi-separable polynomials, then $\phi$ becomes injective.

Let $g\in K[x]$ be a monic semi-separable polynomial such that $g'/g \in L(x)$. Let $\overline{g}\in \bar{L}[x]=\bar{K}[x]$ be a conjugate of $g$ over $L$. Since $g'/g \in L(x)$, it follows that
\[ \phi(\overline{g})=\overline{g}'/\overline{g}=\overline{(g'/g)}=g'/g=\phi(g), \]
By the injectivity of $\phi$ on monic semi-separable polynomials, $\overline{g}=g$ for any conjugate of $g$ over $L$ in $\overline{K}[x]$. Therefore, $g\in L[x]$ ($K/k$ and hence $K/L$ are assumed to be separable extensions throughout this paper).
\end{proof} 

\begin{lemma}\label{lemma:gprime2}
Let $g\in K[x]$ monic, $deg(g)=n$, and let $L$ be a subfield of $K$. Let $p_1,\ldots, p_{2n}\in k$ be distinct elements. If $g'(p_i)/g(p_i)\in L$, $1\leq i \leq 2n$, then $g'/g \in L(x)$. 
\end{lemma}
\begin{proof}
Let $h=g'/g\in K(x)$ and suppose that $h(p_i)\in L$, $1\leq i \leq 2n$. Let $\overline{h}=\overline{g}'/\overline{g}$ be a conjugate of $h$ over $L$. Then \[h(p_i)=\overline{h(p_i)}=\overline{h}(\overline{p_i})=\overline{h}(p_i),\;\;1\leq i \leq 2n.\]
This means that the polynomial $g'\overline{g} - \overline{g}'g$ of degree $<2n$ has $2n$ distinct roots. Hence, $g'\overline{g} - \overline{g}'g=0$ and therefore $h=\overline{h}$, for every conjugate $\overline{h}$ of $h$ over $L$. That is, $g'/g\in L(x)$.
\end{proof}

Consider the following subroutine \texttt{Equations}.

\begin{algorithm}
\caption{\texttt{Equations}.}\label{alg:eq1}
\begin{algorithmic} 
\STATE \textbf{Input:} Subfield factorization $f_1,\ldots, f_r$ of $f$ and an index $i$.
\STATE \textbf{Output:} Set of equations $\mathcal{E}$ whose solutions give the partition $P_i$.
\vspace{0.1cm}
\STATE \text{1. Choose distinct elements $p_1,\ldots, p_{2n}$ of $k$.}
\STATE \text{2. Let $q_j(\alpha):=\sum e_i f_i'(p_j)/f_i(p_j)$, where $q_j(x)\in e_1\cdot k[x]+\cdots +e_r\cdot k[x].$}
\STATE \text{3. Let $\mathcal{E}$ be the system of $k$-linear equations obtained by taking the }
\STATE \text{\;\;\;\; coefficients of $x$ and $\alpha$ of
$\text{rem}(q_j(x), f_i)-q_j(\alpha)=0$, for $j=1,\ldots,2n,$}
\STATE \text{\;\;\;\; where $\text{rem}(q_j(x), f_i)$ is the remainder of the division of $q_j(x)$ by $f_i$.}
\STATE \text{4. \textbf{return} $\mathcal{E}$.}
\end{algorithmic}
\end{algorithm}

This algorithm requires that $k$ has at least $2n$ elements. In practice, however, one often needs very few points to find the partition $P_i$.
By construction, $\mathcal{E}$ has a basis of solutions in $\{0,1\}$-echelon form:

\begin{definition}\label{eq:rref}
A basis of solutions $\{ s_1,\ldots, s_t\}$ of $\mathcal{E}$ is called a \emph{$\{0,1\}$-echelon basis of $\mathcal{E}$} if
\begin{enumerate}
\item $s_i=(s_{i,1}, \ldots, s_{i,r})\in \{0,1\}^r$, $1\leq i \leq t$. 
\item $\sum_{i=1}^t s_i = (1,\ldots, 1)$.
\end{enumerate}
\end{definition}

\begin{remark}\label{rem:echelon}
If a $\{0,1\}$-echelon basis of $\mathcal{E}$ exists, then any reduced echelon basis of $\mathcal{E}$ is automatically a $\{0,1\}$-echelon basis due to the uniqueness of the reduced echelon basis.
\end{remark}

\begin{corollary}\label{cor:e}
Let $\{s_1,\ldots,s_t \}$ be a $\{0,1\}$-echelon basis of $\mathcal{E}$ and let $P_i=\{P^{(1)}, \ldots, P^{(t)}  \}$, where $P^{(l)}=\{j : s_{l,j}=1 \}.$ Then  $P_i$ is the partition of $L_i$.
\end{corollary}
\begin{proof}
If $(e_1,\ldots, e_r)\in\{0,1\}^r$ is a solution of $\mathcal{E}$ then, by Lemmas~\ref{lemma:gprime1} and \ref{lemma:gprime2}, it follows that $g=\prod_{j=1}^r f_j^{e_j} \in L_i[x].$ Thus, the $P_i$-products are in $L_i[x]$. The maximality of $|P_i|$ follows from the fact that $s_1,\ldots, s_t$ form a basis for the solution space of $\mathcal{E}$ and that any vector $(e_1,\ldots, e_r)\in \{0,1\}^r$ such that $\prod_{j=1}^r f_j^{e_j} \in L_i[x]$ is a solution of $\mathcal{E}$. Hence, the partition $P_i$ satisfies Definition~\ref{def:subfieldpart} and therefore, is the partition defined by $L_i$.
\end{proof}

The partition $P_i$ defined by $L_i$ can be found using the following algorithm.

\begin{algorithm}[H]
\caption{\texttt{Partition} (Slow version).}\label{alg:partition_slow}
\begin{algorithmic}
\STATE \textbf{Input:} Subfield factorization $f_1,\ldots, f_r$ of $f$ and an index $i$.
\STATE \textbf{Output:} The partition $P_i$ of $\{1,\ldots, r\}$ defined by $L_{i}$.
\vspace{0.2cm}
\STATE \text{1. Compute $\mathcal{E}$ using subroutine \texttt{Equations}.}
\STATE \text{2. Compute a $\{0,1\}$-echelon basis $\{s_1,\ldots, s_t\}$ of $\mathcal{E}$.}
\STATE \text{3. \textbf{return} $P_i:=\{P^{(1)}, \ldots, P^{(t)}  \}$, where $P^{(l)}$ is as in Corollary \ref{cor:e}.}
\end{algorithmic}
\end{algorithm}

This algorithm, however, does not perform very well in practice. Apart from the (costly) $2n$ polynomial divisions over $K$ in Step~3 of \texttt{Equations}, the system $\mathcal{E}$ is over-determined. The number of linear equations in $\mathcal{E}$ is bounded by $2n^2d_i$, where $d_i=\deg(f_i),$ while the number of variables is $r\leq n$. Furthermore, the coefficients are in $k$ and can be potentially large, while the solutions are $0$-$1$ vectors (to find those, all that is needed are their images modulo a prime number). We address these problems by computing a subset of $\mathcal{E}$ modulo a prime ideal $\mathfrak{p}$.

\begin{definition}\label{def:GoodIdeal}
A {\em good $k$-valuation w.r.t. f } is a valuation $v: k\rightarrow \mathbb{Z}\cup \{\infty\}$ such that if $R_v=\{a\in k \; : \; v(a)\geq 0\}$ and $p_v=\{a\in k\; : \; v(a)>0\}$, then $f\in R_v[x]$, the residue field $\textbf{F}:=R_v/p_v$ is finite, the image $\bar{f}$ of $f$ in $\textbf{F}[x]$ is separable and $\deg(\bar{f})=\deg(f)$. Furthermore, we call an ideal $\mathfrak{p}$ a {\em good $k$-ideal} if $\mathfrak{p} = p_v$, for some good $k$-valuation $v$. 
\end{definition}

If $k=\mathbb{Q}$, then a good $k$-ideal $\mathfrak{p}$ is of the form $(p)$, for some prime number $p$ such that $f\bmod p$ is separable and has the same degree as $f$. The following subroutine returns $\tilde{\mathcal{E}}$: a subset of $\mathcal{E}$ modulo a good $k$-ideal $\mathfrak{p}$.

\begin{algorithm}[H]
\begin{algorithmic}
\caption{\texttt{EquationsModP}.}\label{alg:eq2}
\STATE \textbf{Input:} Subfield factorization $f_1,\ldots, f_r$, an index $i$ and a good $k$-ideal $\mathfrak{p}$.
\STATE \textbf{Output:} $\tilde{\mathcal{E}}$: necessary equations modulo $\mathfrak{p}$ for $e_1,\ldots, e_r$.
\vspace{0.2cm}
\STATE \text{1. Choose $c\in \textbf{F}$ random.}
\STATE \text{2. If $f_j(c)\bmod \mathfrak{p}$ has no inverse, for some $1\leq j \leq r$, go to Step 1.}
\STATE \text{3. Let $q(\alpha):=\sum e_j f_j'(c)/f_j(c)$, where $q(x)\in e_1\cdot\textbf{F}[x]_{<n}+\cdots + e_r\cdot \textbf{F}[x]_{<n}$.}
\STATE \text{4. Let $\tilde{\mathcal{E}}$ be the system of $\textbf{F}$-linear equations obtained by taking the }
\STATE \text{\;\;\;\; coefficients of $x$ and $\alpha$ of %$q_j(x)\equiv q_j(\alpha) \mod f_i$
$\text{rem}(q(x), f_i)-q(\alpha)=0$.}
\STATE \text{5. \textbf{return} $\tilde{\mathcal{E}}$.}
\end{algorithmic}
\end{algorithm}

The inverse in Step~2 is taken in the finite ring $\textbf{F}[\alpha]$, and might not exist. If $\textbf{F}$ is too small in Steps~1-2, and since the solutions are $0$-$1$ vectors, one can compute a finite extension $\tilde{\mathbf{F}}$ of $\textbf{F}$ and compute/solve the system $\tilde{\mathcal{E}}$ over $\tilde{\mathbf{F}}$. 
 
The partition $P_i$ defined by $L_i$ can be computed with the following algorithm.

\begin{algorithm}[H]
\caption{\texttt{Partition}.}\label{alg:partition}
\begin{algorithmic}
\STATE \textbf{Input:} Subfield factorization $f_1,\ldots, f_r$, an index $i$ and a good $k$-ideal $\mathfrak{p}$.
\STATE \textbf{Output:} The partition $P_i$ of $\{1,\ldots, r\}$ defined by $L_{f_i}$.
\vspace{0.2cm}
\STATE \text{1. Compute $\tilde{\mathcal{E}}$ using \texttt{EquationsModP}.}
\STATE \text{2. Compute a $\{0,1\}$-echelon basis $\{s_1,\ldots, s_t\}$ of $\tilde{\mathcal{E}}$ (see Remark~\ref{rem:echelon}).}
\STATE \text{3. \textbf{if} Step 2 fails \textbf{then}}
\STATE \text{4.\;\;\;\;\;\;Compute more equations with \texttt{EquationsModP}.}
\STATE \text{5.\;\;\;\;\;\;Go to Step 2.}
\STATE \text{6. Let $\tilde{P}_i:=\{P^{(1)}, \ldots, P^{(t)}  \}$, where $P^{(l)}$ is as in Corollary \ref{cor:e}.}
\vspace{0.1cm}
\STATE \text{7. Let $\tilde{g}_1,\ldots, \tilde{g}_t$ be the $\tilde{P}_i$-products. \; // }
\STATE \text{8. Let $\mathfrak{q}$ be a good $K_i$-ideal. \;\;\;\;\;\;\;\;\;\;\;\;\;\;\;//}
\STATE \text{9. \textbf{for} $j=1,\ldots, t$ \textbf{do} \;\;\;\;\;\;\;\;\;\;\;\;\;\;\;\;\;\;\;\;\;\;\;\;\;\;// Correctness check (Theorem~\ref{theo:correctPi}).}
\STATE \text{10.\;\;\;\;\; \textbf{if }$\sigma_i(\tilde{g}_j) \not\equiv \tilde{g}_j \bmod \mathfrak{q}$ \textbf{then} \;\;\;\;\;\;\;// }
\STATE \text{11.\;\;\;\;\;\;\;\;\;\;\;\;Go to Step 4. \;\;\;\;\;\;\;\;\;\;\;\;\;\;\;\;\;\;\;\;\;//}
\vspace{0.1cm}
\STATE \text{12. \textbf{return} $\tilde{P}_i$.}
\end{algorithmic}
\end{algorithm}

Our next task is to prove correctness of Algorithm \texttt{Partition}.

\begin{lemma}\label{lemma:modp}
Let $K$ be a field and $f\in K[x]$ monic separable such that $ f = g_1 \cdots g_t = h_1\cdots h_t, $ where $g_j, h_j\in K[x]$ are monic but not necessarily irreducible. Let $\mathfrak{q}$ be a good $K$-ideal. If $g_j \equiv h_j \bmod \mathfrak{q}$, for every $1\leq j \leq t$, then $g_j = h_j$, $1 \leq j \leq t$.
\end{lemma}
\begin{proof}
It suffices to show that for every irreducible factor $q$ of $f$ in $K[x]$, $q\mid g_j$ if and only if, $q \mid h_j$.
Suppose that $q\mid g_j$. Then $q\nmid g_l$, for any $l\neq j$, because $f$ is separable. Moreover, $q$ also does not divide $g_l\bmod \mathfrak{q}$, $l\neq j$, because $f$ is separable modulo $\mathfrak{q}$.
Since $g_l \equiv h_l \bmod \mathfrak{q}$, it follows that $q\nmid h_l \bmod \mathfrak{q}$ and hence, $q\nmid h_l$ over $K$, for all $l \neq j$. But $q$ divides $f = h_1 \cdots h_t$ and since $K[x]$ is a unique factorization domain, it follows that $q\mid h_j$. 
The converse follows similarly. Hence $q\mid g_j$ if and only if, $q\mid h_j$. Since this holds for any irreducible factor $q$ of $f$ in $K[x]$ and $g_j,h_j$ are monic, the equality follows.
\end{proof}
\begin{remark} \label{inclusions}
When choosing the ideal $\mathfrak{q}$ we have to make sure that denominators of coefficients of $g_j$ and $h_j$ are not elements of $\mathfrak{q}$, otherwise the equation $g_j \equiv h_j \bmod \mathfrak{q}$ would return an error message. For $k=\mathbb{Q}$ and assuming $f$ monic, the following inclusions \[ \mathbb{Z}[\alpha] \subseteq \mathcal{O}_K \subseteq \frac{1}{f'(\alpha)}\cdot\mathbb{Z}[\alpha] \subseteq \frac{1}{\text{disc}(f)}\cdot \mathbb{Z}[\alpha],  \] where $\mathcal{O}_K$ is the ring of integers of $K$ and $\text{disc}(f)$ is the discriminant of $f$, and the fact that any factor of $f$ over $K$ is in $\mathcal{O}_K[x]$, by Gauss' Lemma, imply that it is enough to choose $\mathfrak{q}$ such that $\text{disc}(f) \not \equiv 0 \bmod \mathfrak{q}$.
\end{remark}

\begin{theorem}\label{theo:correctPi}
If Algorithm~\rm{\texttt{Partition}} \it finishes, the output $\tilde{P}_i$ is the partition defined by the principal subfield $L_{f_i}$.
\end{theorem}
\begin{proof}
Let $f_1,\ldots, f_r \in K[x]$ be a subfield factorization of $f$ and let $P_i$ be the (correct) partition defined by $L_{f_i}$. By reducing the number of equations and solving the linear system $\tilde{\mathcal{E}}$ over $\textbf{F}$, the partition $\tilde{P}_i$ at Step~6 of Algorithm \texttt{Partition} is either $P_i$ or a proper refinement of $P_i$. Let $K_i:=K[y]/\left\langle f_i(y) \right\rangle$ and define $\sigma_i:K \rightarrow K_i$, $\sigma_i(\alpha)= y+(f_i)$. With this notation, Equation~(\ref{eq:Lg}), with $g=f_i$, becomes \begin{equation}
 L_i=L_{f_i} =\left\{ h(\alpha) : h\in k[x]_{<n},\; \sigma_i(h(\alpha)) = h(\alpha)\right\}.
 \end{equation}
Let $\tilde{P}_i = \{\tilde{P}^{(1)},\ldots, \tilde{P}^{(t)} \}$ be the partition defined by the $\{0,1\}$-echelon basis of $\tilde{\mathcal{E}}$ (if there is no such basis, the algorithm computes more equations in Step~4) and let $\tilde{g}_1,\ldots ,\tilde{g}_t\in K[x]$ be the $\tilde{P}_i$-products. In order for $\tilde{P}_i$ to satisfy Definition~\ref{def:subfieldpart} (i.e., $\tilde{P}_i = P_i$), it suffices to show that $\tilde{g}_j \in L_i[x]$, for $1\leq j \leq t$ (the maximality of $t$ will follow from the fact that $\tilde{P}_i$ refines $P_i$). That is, with the notations above, we need to show that $\sigma_i(\tilde{g}_j)=\tilde{g}_j$, where $\sigma_i$ acts on $\tilde{g}_j\in K[x]$ coefficient-wise. Since \begin{equation}
 \tilde{g}_1 \cdots \tilde{g}_t = f =\sigma_i(f) = \sigma_i(\tilde{g}_1)\cdots \sigma_i(\tilde{g}_t)
 \end{equation} over $K_i$, we can choose a good $K_i$-ideal $\mathfrak{q}$ and use Lemma~\ref{lemma:modp} to show that we only need to verify whether
\begin{equation}\label{eq:sigma}
\sigma_i(\tilde{g}_j)\equiv \tilde{g}_j \bmod \mathfrak{q}.
\end{equation}
Hence, if $\{s_1,\ldots, s_t\}$ is a $\{0,1\}$-echelon basis of $\tilde{\mathcal{E}}$ and if (\ref{eq:sigma}) holds for $j=1,\ldots, t$, then $\tilde{P}_i$ is the partition defined by $L_i$.

If $K_i$ is not a field, we cannot directly apply Lemma~\ref{lemma:modp}. Let $f_i=f_{i_1} \cdots f_{i_s}$, with $f_{i_m}\in K[x]$ irreducible, $m=1,\ldots, s$. Let $K_{i_m}:=K[y]/\left\langle f_{i_m}(y) \right\rangle$ and define $\sigma_{i_m}:K\rightarrow K_{i_m}$ as above. Since $f$ is separable, it follows that $\tilde{g}_j\in L_i[x]$ if and only if, $\sigma_{i_m}(\tilde{g}_j)= \tilde{g}_j$, $m=1,\ldots, s$. To use Lemma~\ref{lemma:modp}, we would need $\mathfrak{q}$ to be a good $K_{i_m}$-ideal. However, we can view $K_{i_m}=K[\alpha_{i_m}]$, where $\alpha_{i_m}$ is a root of $f_{i_m}$, and choose $\mathfrak{q}$ to be a good $K$-ideal. Thus, by Lemma~\ref{lemma:modp} (with $\sigma_{i_m}$ instead of $\sigma_i$ in the argument above and $\mathfrak{q}$ a good $K$-ideal), it follows that $\sigma_{i_m}(\tilde{g}_j)=\tilde{g}_j$, if and only if, $\sigma_{i_m}(\tilde{g}_j) \equiv \tilde{g}_j \bmod \mathfrak{q}$, $m=1,\ldots, s$. Since $f\bmod \mathfrak{q}$ is separable, this is equivalent to $\sigma_i(\tilde{g}_j)\equiv \tilde{g}_j \bmod \mathfrak{q}$. That is, if the $\tilde{P}_i$-products satisfy Equation~(\ref{eq:sigma}), then $\tilde{P}_i$ is the partition of $L_i$.
\end{proof}

We were not able to bound the number of calls to \texttt{EquationsModP} when computing the partition $P_i$. However, based on our experiments for $k=\mathbb{Q}$, the average number of calls to \texttt{EquationsModP} appears to be bounded by a constant (in fact, this number never exceeded $3$ in our examples). For this reason, we shall assume that the number of calls to \texttt{EquationsModP} is $\mathcal{O}(1)$.

\begin{theorem} \label{THM_P_i}
Assuming that the number of calls to \rm\texttt{EquationsModP} \it is bounded by a constant, when $k=\mathbb{Q}$, the number of CPU operations for computing $P_i$ is
\[\tilde{\mathcal{O}}(n(n^2+n\log \|f\|+d_i r^{\omega-1}))\]
where we omit $\log p$ factors in $\tilde{\mathcal{O}}$ notation (to bound $p$ see Remark~\ref{RemarkP} below).
\end{theorem}
\begin{proof}
To prove this we first bound the cost of calling Algorithm \texttt{EquationsModP}. The integer coefficients of $f'(\alpha) f_j\in \mathbb{Z}[\alpha][x]$ can be bounded by $n4^n\|f\|^2$ (see Lemma~\ref{lemma:bound2}, Appendix~\ref{App:C}). Hence, computing $f'(\alpha)f_j(c)$ modulo $p$, for $1\leq j \leq r$, has a cost of $\tilde{\mathcal{O}}(n^2(n+\log\|f\|))$ CPU operations.
Then 
the divisions $f'_j(c)/f_j(c)\bmod p$ in Step~3 of \texttt{EquationsModP} can be executed with $\tilde{\mathcal{O}}(rn)$
CPU operations and Step~4 has a cost of $\tilde{\mathcal{O}}(rn^2)$ CPU operations.
One call of \texttt{EquationsModP} has a cost of $\tilde{\mathcal{O}}(n^2(n+\log \|f\|))$ CPU operations. In our experiments, the number of calls to
algorithm \texttt{EquationsModP} from Algorithm \texttt{Partition} was never more than 3. Usually 1 call sufficed to find the partition $P_i$. In this case, the system $\tilde{\mathcal{E}}$ has at most $nd_i$ equations in $r$ variables, where $d_i=\deg(f_i)$. Hence, a solution basis can be found with $\tilde{\mathcal{O}}(nd_i r^{\omega-1})$ CPU operations. The cost of Steps 7-11 in Algorithm \texttt{Partition} is given by the cost of computing the polynomials $\tilde{g}_j,1 \leq j \leq t$, which can be done with at most $r-1$ polynomial multiplications in $\mathbb{F}_p(\alpha)[x]$, and the cost of $nt$ divisions in $\mathbb{F}_p[x]$.
\end{proof}

\begin{remark} \label{RemarkP}
One could design the algorithm to work with any $p$ for which $f$ is separable mod $p$, with $p$ not dividing the leading coefficient.
Then $\log p$ can be bounded as $\mathcal{O}( {\rm log}(n+||f||) )$ by
Equation~(3.9) in \cite{LLL}.  But it is best to select $p$ for which $f$ has a root mod $p$.
The probability that $f$ has a root mod a random prime $p$
is asymptotically at least $1/n$ by Chebotarev's density theorem.
With the (unproven, but true in experiments) assumption that this probability is not much smaller for small $p$,
the expected size for $\log p$ is still bounded by $\mathcal{O}( {\rm log}(n+||f||) )$.
\end{remark}

\section{A General Algorithm and Generators}\label{sec:4}
In this section, we combine the ideas of Sections~\ref{sec:2} and \ref{sec:3} and give a general algorithm for computing all subfields of $K/k$. Given a partition $P_L$, we also present an algorithm in Subsection~\ref{ssec:generators} for computing a set of generators for $L$.

\subsection{The \texttt{Subfields} Algorithm}

The algorithm \texttt{Subfields} below returns a set of partitions representing every subfield of $K/k$. This is particularly useful if one wants the subfield lattice of the extension $K/k$. On the other hand, these partitions and the subfield factorization of $f$ allow us to give the subfield polynomial of each subfield of $K/k$ in (partially) factored form. 

\begin{algorithm}[H]
\caption{\texttt{Subfields}.}\label{alg:general}
\begin{algorithmic}
\STATE \textbf{Input:} An irreducible squarefree polynomial $f\in k[x]$.
\STATE \textbf{Output:} A data structure that lists all subfields of $K/k$ (by giving their  \newline subfield polynomial in factored form).

\vspace{0.2cm}
\STATE \text{1. Compute a subfield factorization $f_1\cdots f_r$ of $f$ in $K[x]$.}
\STATE \text{2. \textbf{for} $i=1,\ldots, r$ \textbf{do}}
\STATE \text{3. \;\;\;\; Compute the partition $P_i$ using algorithm $\texttt{Partition}$.}
\STATE \text{4. $S_0:=\{P_1,\ldots, P_r\}$.}
\STATE \text{5. $S:=S_0$.}
\STATE \text{6. \textbf{for} $P$ in $S_0$ \textbf{do}}
\STATE \text{7. \;\;\;\; $S:=S \cup \{P \vee Q : Q \in S\}$.}
\STATE \text{8. \textbf{return} $S$ and $[f_1,\ldots, f_r]$.}
\end{algorithmic}
\end{algorithm}

Next, we analyze the complexity of Algorithm \texttt{Subfields} for the case $k=\mathbb{Q}$.

\begin{theorem}\label{theo:principal}
Let $m$ be the number of subfields of $K/k$. 
Under the assumptions in Theorem~\ref{THM_P_i} and~\ref{ComplexityMethod2}, when $k=\mathbb{Q}$, Algorithm \rm{\texttt{Subfields}} \it performs $\tilde{\mathcal{O}}(rn^5(n+\log \|f\|_2)^2+mr^2)$ CPU operations, where $n$ is the degree of the extension $K/k$ and $r$ is the number of factors in the subfield factorization.
\end{theorem}
\begin{proof}
In Step~1 we have to compute a subfield factorization of $f$ over $K$. Using Algorithm \texttt{SubfFact} presented in Section~\ref{sec:5}, this step can be executed with $\tilde{\mathcal{O}}(rn^5(n+\log \|f\|)^2)$ CPU operations. In Steps 2-3 we have to compute $r$ partitions, where each partition can be computed with an expected number of $\tilde{\mathcal{O}}(n(n^2+n\log \|f\|+d_i r^{\omega-1}))$ CPU operations, where $d_i$ is the degree of $f_i$. Finally, the set $S$ never has more than $m$ elements, and the set $S_0$ has at most $r$ elements. Therefore, the number of times we compute $P \vee Q$ is bounded by $r m$. Since the cost of each partition join is $\tilde{\mathcal{O}}(r)$, the cost of Steps 6-7 is given by $ \tilde{\mathcal{O}}(mr^2)$ CPU operations. 
\end{proof}

Steps $6$ and $7$ find the partitions of all subfields from the partitions of the principal subfields of $K/k$. However, the same subfield might be computed several times. A more elaborate way to compute all subfields from the principal subfields, avoiding the computation of the same subfield several times, is given in \cite{HKN} (though the bound for the number of intersections/joins is the same).

Since the number of subfields $m$ is not polynomially bounded, the theoretical worst-case complexity is dominated by the cost of all intersections of the principal subfields $L_1,\ldots, L_r$. Since each subfield is represented by a partition and the intersection of subfields can be computed by joining partitions, we were able to improve the theoretical complexity. Moreover, computing all subfields using partitions only contributes to a small percentage of the total CPU time.

\subsection{From a Partition to a Subfield}\label{ssec:generators}

In addition to returning the subfield lattice (in terms of partitions), one can also compute generators for any subfield of $K/k$. Let $f_1,\ldots, f_r$ be a subfield factorization and let $L_1,\ldots, L_r$ be the principal subfields. Given a partition $P_L$, corresponding to a subfield $L$ of $K/k$, one can find a set of generators of $L$ by expanding the subfield polynomial $g_L$ of $L$ (recall that $g_L=\prod_{j\in P_L^{(1)}} f_j$) and taking its coefficients (see Theorem~\ref{theo:equiv}). This gives us the following algorithm.

\begin{algorithm}[H]
\caption{\texttt{Generators} (Slow version).}
\begin{algorithmic}
\STATE \textbf{Input:} Subfield factorization $f_1,\ldots, f_r$ of $f$ and the partition $P_L$.
\STATE \textbf{Output:} A set of generators of the subfield $L$ of $K/k$.
\vspace{0.2cm}
\STATE \text{1. Compute $g_L:=\prod_{j\in P_L^{(1)}} f_j$.}
\STATE \text{2. \textbf{return} the set of coefficients of $g$.}
\end{algorithmic}
\end{algorithm}

However, expanding the subfield polynomial can be an expensive task, especially when $g_L$ has high degree. Alternatively, one can compute only a few (easy to compute) coefficients of $g_L$ (for example, if $d=\deg(g_L)$, then the coefficient of $x^{d-1}$ and the trailing coefficient are easy to compute from the partial factorization of $g_L$) or one can compute $g_L(c)=\prod_{i\in P_L^{(1)}} f_i(c)$, for $c\in k$, for as many $c$ as we want. Let us denote by \texttt{NextElem}( ) a procedure that returns elements of $L$. What we need now is a practical criterion that tells us when a set of elements of $L$ generates $L$.

\begin{theorem}\label{theo:gen}
Let $\beta_1,\ldots, \beta_s \in L$ and let $P_L$ be the partition defined by $L$. Then $L=k(\beta_1,\ldots, \beta_s)$ if and only if, for any $j\notin  P_L^{(1)}$ there exists $l\in \{1,\ldots, s\}$ such that $\beta_l\notin L_j$. 
\end{theorem}
\begin{proof}
Notice that $L\cap L_j \subsetneq L$, for any $j\notin P_L^{(1)}$. Hence, if there exists some $j\notin P_L^{(1)}$ such that $\beta_i \in L_j$, for every $\beta_i$, then $ k(\beta_1\ldots,\beta_s)\subseteq L\cap L_j \subsetneq L. $

Conversely, let $\beta_1\ldots, \beta_s\in L$ be such that for any $j\notin P_L^{(1)}$, there exists $\beta_i$ such that $\beta_i\notin L_j$. Let $\tilde{L}:=k(\beta_1\ldots, \beta_s)$ and suppose that $\tilde{L}\subsetneq L$. Let $P_{\tilde{L}}$ be the partition defined by $\tilde{L}$. By Lemma~\ref{lemma:subset} we have $P_L^{(1)}\subsetneq P_{\tilde{L}}^{(1)}$ and hence, there exists $j\in P_{\tilde{L}}^{(1)}$ such that $j\notin P_L^{(1)}$ and $\beta_i \in L_j$, for any $i\in P_L^{(1)}$, which is a contradiction. Therefore, $L=k(\beta_1\ldots, \beta_s).$
\end{proof}

Recall that for any element $\beta\in K$, there exists $g(x)\in k[x]_{<n}$ such that $\beta=g(\alpha)$ and that $\beta \in L_j$ if and only if, $g(x) \equiv g(\alpha) \bmod f_j$. To show that $\beta\notin L_j$, it suffices to show that \[g(x) \not \equiv g(\alpha) \bmod (f_j,\mathfrak{p}),\] where $\mathfrak{p}$ is as in Definition \ref{def:GoodIdeal}. Theorem~\ref{theo:gen} allows us to write an algorithm for computing a set of generators of $L$.

\begin{algorithm}
\caption{\texttt{Generators}.}
\begin{algorithmic}
\STATE \textbf{Input:} Subfield factorization $f_1,\ldots, f_r$ of $f$ and the partition $P_L$.
\STATE \textbf{Output:} A set of generators of the subfield $L$ of $K/k$.
\vspace{0.2cm}
\STATE \text{1. $S:=\emptyset$.}
\STATE \text{2. $J:=\{1,\ldots,r \} - P_L^{(1)}$.}
\STATE \text{3. $\beta:=\texttt{NextElem}(\;)$, where $\beta=g(\alpha), \text{ for some }g(x)\in k[x]_{<n}$.}
\STATE \text{4. $S:=S \cup \{\beta\}$.}
\STATE \text{5. \textbf{for} $j\in J$ \textbf{do}}
\STATE \text{6. \;\;\;\; \textbf{if} $g(x)\not \equiv g(\alpha)\bmod (f_j,\mathfrak{p})$ \textbf{then} $J:=J-\{j\}$.} 
\STATE \text{7. \textbf{if} $J\neq \emptyset$ \textbf{then} Go to Step 3 \textbf{else} \textbf{return} $S$.}
\end{algorithmic}
\end{algorithm}

\begin{theorem}
The output of Algorithm \rm{\texttt{Generators}} is a set $S\subseteq L$ which generates $L$.
\end{theorem}
\begin{proof}
If $g(x)\not \equiv g(\alpha)\bmod (f_j,\mathfrak{p})$ in Step 6, then $g(x)\not \equiv g(\alpha)\bmod f_j $ and hence, $g(\alpha)\notin L_j$. If $S$ is the output of Algorithm \texttt{Generators}, then for any $j\notin P_L^{(1)}$, there exists $\beta\in S$ such that $\beta\notin L_j$. By Theorem~\ref{theo:gen}, $S$ generates $L$.
\end{proof}

Algorithm \texttt{Generators}, as it is stated, is not guaranteed to finish. If the algorithm has not found a generating set after a certain number of elements computed, one could compute the subfield polynomial and return its coefficients.

\section{The Number Field Case}\label{sec:5}

In this Section, $k = \mathbb{Q}$ and $K = \mathbb{Q}(\alpha)$ with $f \in \mathbb{Z}[x]$ irreducible, and $f(\alpha)=0$.
A subfield factorization can be obtained by fully factoring $f$ in $K[x]$ (e.g. with \cite{Trager} or \cite{belabas}).
This Section shows how one can find a subfield factorization with LLL, without fully factoring $f$ over $K$.
With a fast implementation of \cite{belabas} (see Section~6.1 for timings), the reader may simply want to use that, and skip this Section (this Section is still useful
to get a complexity estimate missing in \cite{belabas}).

 We start by choosing a prime $p$ such that $p$ does not divide the leading coefficient of $f \in \mathbb{Z}[x]$, $f\bmod p$ is separable and has at least one linear factor in $\mathbb{F}_p[x]$, denoted as $\bar{f}_1$.
Let $\hat{K} = \mathbb{Q}_p$ be the field of $p$-adic numbers. The factorization $\bar{f}_1, \ldots, \bar{f}_{\hat{r}}$ of $f\bmod p$ lifts to a factorization $\hat{f}_1 \cdots \hat{f}_{\hat{r}}$ of $f$ into irreducible factors over $\mathbb{Q}_p$, with $\hat{f}_1$ linear. We can only compute $p$-adic factors with finite accuracy. For $i=1,\ldots, \hat{r}$ and a positive integer $a$, let $\hat{f}_i^{(a)}\in \mathbb{Z}[x]$ be an approximation of $\hat{f}_i$ with accuracy $a$, that is, $\hat{f}_i^{(a)}\equiv \hat{f}_i \bmod p^a$. % (Hensel lifting).
 By mapping $\alpha\in \mathbb{Q}(\alpha)$ to the root $\hat{\alpha}$ of $\hat{f}_1$ in $\mathbb{Q}_p$, we can view $K=\mathbb{Q}(\alpha)$ as a subfield of $\hat{K}=\mathbb{Q}_p$.  
 
 For $g\in \mathbb{Q}(\alpha)[x]$, we will denote by $\bar{g}\in \mathbb{F}_p[x]$, the image of $g$ under the map $\alpha \rightarrow \bar{\alpha}$, where $\bar{\alpha}$ is the root of $\bar{f}_1$, and by $\hat{g}\in \mathbb{Q}_p[x]$, the image of $g$ under the map $\alpha \rightarrow \hat{\alpha}$, where $\hat{\alpha}$ is the root of $\hat{f}_1$. Furthermore, for $g,h\in \mathbb{Q}(\alpha)[x]$, we denote by $\gcdp(g,h)$ the $\gcd$ of the images $\bar{g}$ and $\bar{h}$ over $\mathbb{F}_p$.

%\subsection{Computing a Subfield Factorization}
% Let $p$ be a prime number as mentioned above and
% Let $\bar{f}_1,\ldots, \bar{f}_{\hat{r}}$ be the irreducible factorization of $f\bmod p$, with $\bar{f}_1$ linear.
Fix $i\in \{1,\ldots, \hat{r}\}$. As shown in \cite{HKN}, one can use LLL to compute linearly independent algebraic numbers $\beta_1,\ldots, \beta_{m_i}\in \mathbb{Q}(\alpha)$ which are likely to form a $\mathbb{Q}$-basis of $L_i$ (it is only guaranteed that $L_i\subseteq \mathbb{Q}\cdot \beta_1+\cdots+ \mathbb{Q}\cdot \beta_{m_i}$ as $\mathbb{Q}$-vector spaces). The idea of the following algorithm is to use this basis to compute the subfield polynomial $g_{L_i}$ of $L_i$ and construct a subfield factorization iteratively. 

\begin{algorithm}[H]
\caption{\texttt{PartialSubfFact}}
\begin{algorithmic}
\STATE \textbf{Input:} A $\mathbb{Q}$-basis $\beta_1,\ldots, \beta_{m_i}$ of some $V$ such that $L_i\subseteq V$ and a partial factorization $g_1,\ldots, g_s$ of $f$ over $\mathbb{Q}(\alpha)$.
\STATE \textbf{Output:} A partial factorization $G_1,\ldots, G_S$ of $f$ over $\mathbb{Q}(\alpha)$, with $s\leq S$, and such that $g_{L_i} \in \{G_1,\ldots, G_S\}^\pi$ or \textbf{Error}.
\vspace{0.2cm}
\STATE \text{ 1. Let $SF:=\{g_1,\ldots, g_s\}$ and let $T\subseteq k$ finite.}
\STATE \text{ 2. Let $\beta$ be a random $T$-combination of $\beta_1,\ldots, \beta_{m_i}$.}
\STATE \text{ 3. Let $H:=h(x)-h(\alpha)$, where $h(x)\in \mathbb{Z}[x]_{<n}$ and $h(\alpha)=\beta$. }
\STATE \text{ 4. Compute $g_0:=\gcdp(f,H)$ in $\mathbb{F}_p[x]$.}
\STATE \text{ 5. \textbf{if} $\deg(g_0)\cdot m_i \neq n$ \textbf{then} go to Step 2.}
\STATE \text{ 6. \textbf{for} $j=1,\ldots,s$ \textbf{do}}
\STATE \text{ 7. \;\;\;\; Compute $g:=\gcdp(g_j,H)$ in $\mathbb{F}_p[x]$.}
\STATE \text{ 8. \;\;\;\; \textbf{if } $ 0< \deg(g)< \deg(g_j)$ \textbf{then}}
\STATE \text{ 9. \;\;\;\;\;\;\;\;\;\;\;\; Compute $G:=\gcd(g_j, H)$ in $\mathbb{Q}(\alpha)[x].$}
\STATE \text{ 10. \;\;\;\;\;\;\;\;\;\; \textbf{if} $\bar{f}_i \mid \bar{g_j}$ but $ \bar{f}_i \nmid \bar{G}$ \textbf{then} \textbf{return} \textbf{Error}.}
\STATE \text{ 11. \;\;\;\;\;\;\;\;\;\; $SF:=(SF - \{g_j\})\cup \{G, g_j/G \}$.}
\STATE \text{ 12. \textbf{return} $SF$}
\end{algorithmic}
\end{algorithm}

When $\beta_1,\ldots, \beta_{m_i}$ is not a $\mathbb{Q}$-basis of $L_i$, Step 5 might give rise to an infinite loop. Otherwise, $\deg(g_0)\cdot m_i \neq n$ when the random element $\beta$ is not a generator of $L_i$, which happens with probability at most $(m_i-1)|T|^{m_i(1-q)/q}$, where $q$ is the smallest prime that divides $m_i$ (see Appendix \ref{sec:prob}).
To prove the correctness of Algorithm \texttt{PartialSubfFact}, we use the following remark.

\begin{remark}\label{remark:division}
As a consequence of Lemma~\ref{lemma:modp}, if $g,h\in \mathbb{Q}(\alpha)[x]$ are factors of $f$ then one can quickly verify whether or not $h\mid g$ by checking whether the image of $h$ in $\mathbb{F}_p[x]$ divides the image of $g$ in $\mathbb{F}_p[x]$. The same holds for deciding when $\gcd(g,h)\in \mathbb{Q}(\alpha)[x]$ is trivial or not.
\end{remark}

\begin{lemma}\label{lemma:q=li}
If Algorithm \texttt{PartialSubfFact}  does not end in an error message, then the input $\beta_1, \ldots, \beta_{m_i}$ is a basis of $L_i$, and moreover, $L_i=\mathbb{Q}(\beta)$, with $\beta$ from Step~2.  If Step~10 returns an error message, then $\beta_1, \ldots, \beta_{m_i}$ is not a basis of $L_i$.
\end{lemma}

\begin{proof}
Let $g_{L_i}$ be the subfield polynomial of $L_i$ and let $g_\beta$ be the subfield polynomial of $\mathbb{Q}(\beta)$ (see Theorem~\ref{theo:equiv}). Let $g_0\in \mathbb{F}_p[x]$ as in Step~4. It follows that \begin{equation}\label{eqp1}
\deg(g_0)\geq \deg(\bar{g}_\beta)= \deg(g_\beta).
\end{equation} 

Furthermore, since $L_i \subseteq V$ as $\mathbb{Q}$-vector spaces, we have $\dim(V) \geq \dim(L_i)$. But $\dim(L_i) = n/\deg(g_{L_i})$ and $\dim(V) = m_i$. Hence, \begin{equation} \label{eqp2}
\deg(g_{L_i})\geq n/m_i.
\end{equation}

If Step~5 does not generate an infinite loop (in which case the algorithm should return an error message), then $\deg(g_0)\cdot m_i=n$ and hence, Equations $(\ref{eqp1})$ and $(\ref{eqp2})$ tell us that \begin{equation}\label{eqdeg}
\deg(g_{L_i}) \geq n/m_i = \deg(g_0) \geq \deg(g_\beta).
\end{equation} 

Now suppose that Step~10 did not return an error message. Since $f$ is separable modulo $p$, there is only one index $I$, $1\leq I \leq s$, such that $\bar{f}_i\mid \bar{G}$, where $G=\gcd(g_I, H)$. If $F$ is the irreducible factor of $f$ over $\mathbb{Q}(\alpha)$ such that $\bar{f}_i \mid \bar{F}$, then using Remark~\ref{remark:division} one can show that $F \mid G \mid g_\beta$ and hence, $\mathbb{Q}(\beta)=L_{g_\beta}\subseteq L_{F}$. On the other hand, if $\hat{f}_i$ is the $p$-adic factor of $f$ which reduces to $\bar{f}_i$ modulo $p$, then $\hat{f}_i \mid \hat{F}$ and hence, $L_{F}\subseteq L_{\hat{F}} \subseteq L_{\hat{f}_i} = L_i$. Therefore $\mathbb{Q}(\beta) \subseteq L_i$ and hence, $g_{L_i}\mid g_\beta$, by Lemma~\ref{lemma:subset}. Therefore, by Equation~(\ref{eqdeg}), we have $g_{L_i}=g_\beta=\gcd(f,H)$ and hence, \[L_i=\mathbb{Q}(\beta)=V.\]

This also shows that the polynomials $g$ in Step~7 and $G$ in Step~9 have the same degree. If the algorithm does return an error message in Step~10, then $\bar{f}_i\mid \bar{g}_I$ but $\bar{f}_i\nmid \bar{G}$. Hence $F\nmid G$ and since $F\mid g_I$, it follows that $F\nmid H=h(x)-h(\alpha)$. By looking at the images over the $p$-adic numbers, we have $\hat{f}_i\mid \hat{F}\nmid \hat{H}$, which means that $h(\alpha)=\beta\notin L_i$ and hence, $\beta_1,\ldots, \beta_{m_i}$ is not a basis of $L_i$.
\end{proof}

\begin{theorem}
Let $g_{L_i}$ be the subfield polynomial of $L_i$. Given a $\mathbb{Q}$-basis of $V\supseteq L_i$ and a (partial) factorization of $f$, Algorithm \rm{\texttt{PartialSubfFact}} \it returns a (partial) factorization $G_1,\ldots, G_S$ of $f$ such that $g_{L_i} \in \{G_1,\ldots, G_S\}^\pi$ or an error message.
\end{theorem}
\begin{proof}

 If the algorithm does not return an error message, then by Lemma~\ref{lemma:q=li} it follows that $g_{L_i}=\gcd(f, H)$. Hence, by computing the $\gcd$ of $H$ with the partial factorization of $f$ and updating the set $SF$ (Step~11), it follows that the output $SF$ in Step~12 is such that $g_{L_i} \in SF^\pi$.
\end{proof}

Different bases for $\mathbb{Q}(\alpha)$ give different bounds on the bit-size of $\beta_1,\ldots, \beta_{m_i}$. While the standard basis $\{1,\alpha,\ldots, \alpha^{n-1}\}$ simplifies implementation, the \textit{rational univariate representation basis} $\{1/f'(\alpha),\ldots, \alpha^{n-1}/f'(\alpha)\}$ can improve running times and provide better complexity results, see \cite{Alonso} and \cite{Dahan}.

Besides giving better bounds, there are more advantages in using the rational univariate representation basis. For example, if $g$ is a monic factor of $f$ in $\mathbb{Q}(\alpha)[x]$, then $f'(\alpha)g\in \mathbb{Z}[\alpha][x]$ (see~\cite{isomorphisms} or Remark~\ref{inclusions}). This allows us to make simplifications in a general algorithm for computing $\gcd$'s in $\mathbb{Q}(\alpha)[x]$, giving better complexity results. See Appendix~\ref{App:C}.

\begin{remark}\label{rem:ratiorep}
Suppose that $\beta_1,\ldots, \beta_{m_i}$ is a $\mathbb{Q}$-basis of $V\supseteq L_i$. Let $\beta$ be a random $T$-combination of $\beta_1,\ldots, \beta_{m_i}$ and let $b_0,\ldots, b_{n-1}\in \mathbb{Z}$ be such that $\beta=\sum b_j \frac{\alpha^j}{f'(\alpha)}$. If $\tilde{h}(x)=\sum b_j x^j\in \mathbb{Z}[x]$, then one should define $H(x)$ as $\tilde{h}(x)f'(\alpha) - \tilde{h}(\alpha) f'(x)\in \mathbb{Z}[\alpha][x]$ in Step~3 of Algorithm \rm \texttt{PartialSubfFact}. 
\end{remark}

\begin{lemma}\label{lemma:complPSF}
Given a $\mathbb{Q}$-basis of $V\supseteq L_i$ (computed in the rational univariate representation basis) and a partial factorization $g_1,\ldots, g_s $ of $f$, the number of CPU operations for running Algorithm \rm{\texttt{PartialSubfFact}} is bounded by
% expected\footnote{If $V \neq L_i$, which can happen if the value of $a$ in Algorithm~\texttt{SubfFact} is too low,} to be
%   -->  This comment should not be here, but should be elsewhere.
\[\tilde{\mathcal{O}}(n^3(r+\log \|f\|_2)). \]

\end{lemma}

\begin{proof}
The cost of Steps~4 and~7 is less than the cost of Step~9. The cost of the division $g_j/G$ in Step~11 is similar to the cost of the $\gcd$ in Step~9 (this division can be computed by dividing the images of $g_j$ and $G$ in $\mathbb{F}_p(\alpha)[x]$ and then Chinese remaindering). Since $f$ is separable modulo $p$, there is only one $g_I$ such that $\bar{f}_i\mid \bar{g}_I$ and if \begin{equation}\label{eq:gigcd}\bar{f}_i \mid \bar{G},\text{ where } G=\gcd(g_I, H),\end{equation} then, by the proof of Lemma~\ref{lemma:q=li}, we have \[\deg(\gcd(g_j, H))=\deg(\gcdp(g_j, H)), \text{ for any }1\leq j \leq s\] and hence, when computing $\gcd(g_j, H)$, $j\neq I$, we can skip the trial divisions in the modular $\gcd$ algorithm (see~\cite{modular} and Appendix~\ref{App:C}). That is, we have one $\gcd$ computation with trial divisions, which costs $\tilde{\mathcal{O}}(n^3\log \|f\|_2)$ CPU operations, and if (\ref{eq:gigcd}) holds, then we can skip the trial divisions in the remaining $\gcd$'s, where each such $\gcd$ costs $\tilde{\mathcal{O}}(n^2(n+\log\|f\|_2))$ CPU operations (see Appendix~\ref{App:C}). Furthermore, each division test in Step 10 costs $\tilde{\mathcal{O}}(n\log p)$ CPU operations. The result follows by omitting $\log p$ terms and the fact that $s\leq n$.
\end{proof}

A general description of the algorithm to compute a subfield factorization of $f$ over $\mathbb{Q}(\alpha)$ is given below.

\begin{algorithm}
\caption{\texttt{SubfFact}.}\label{alg:general2}
\begin{algorithmic}
\STATE \textbf{Input:} A squarefree irreducible polynomial $f\in \mathbb{Z}[x]$.
\STATE \textbf{Output:} A subfield factorization of $f$ over $\mathbb{Q}(\alpha)$ (see Definition \ref{def:subfact}).
\vspace{0.2cm}
\STATE \text{1. \;Let $p$ prime for which $\bar{f}\in \mathbb{F}_p[x]$ is separable, has a linear factor} 
\STATE \text{ \;\;\;\;and the same degree as $f$.}
\STATE \text{2. \;Compute the irreducible factorization $\bar{f_1},\ldots, \bar{f}_{\hat{r}}$ of $\bar{f}\in \mathbb{F}_p[x].$ }
\STATE \text{3. \;$SF_0 :=\{x-\alpha, f/(x-\alpha)\}$. }
\STATE \text{4. \;\textbf{for} $i=1,\ldots, \hat{r}$ \textbf{do}}
\STATE \text{5. \;\;\;\;\;\;Hensel Lift $\bar{f_1},\ldots, \bar{f}_{\hat{r}}$ to a factorization $f_1^{(a)}, \ldots, f_{\hat{r}}^{(a)}$ of $f$ mod $ p^a$,}
\STATE \text{ \;\;\;\;\;\;\;\;  for appropriate $a$ (starting $a$ is the same as in \cite{HKN}).}
\STATE \text{6. \;\;\;\;\; Use LLL to compute a basis $\beta_1,\ldots, \beta_{m_i}$ of some $V\supseteq L_i$ (See \cite{HKN}).}
\STATE \text{7. \;\;\;\;\; $SF_i$:=\texttt{PartialSubfFact}($\{\beta_1,\ldots, \beta_{m_i}\},SF_{i-1})$. }
\STATE \text{8. \;\;\;\;\; If $SF_i =\text{\textbf{Error}}$, increase the lifting precision $a$, go to Step 5.}
\STATE \text{9. \;\textbf{return} $SF_{\hat{r}}$.}
\end{algorithmic}

\end{algorithm}

\begin{theorem}
\label{ComplexityMethod2}
Assuming a prime $p$ of suitable size (see Remark~\ref{RemarkP}) is found, and assuming the value of $a$ from \cite{HKN} is large
	% enough\footnote{In experiments this
	% is always the case, but even if a higher value of $a$ is needed, the impact on the complexity and CPU time is very small.
	% By using gradual sub-lattice reduction \cite{HA}, the time spent in LLL (which dominates the CPU time) is independent of $a$
	% (the cost of constructing the LLL-inputs does depend on $a$). To prove an upper bound for $a$, bound the coefficients of $H(x)$ in Remark~\ref{rem:ratiorep} by multiplying
	% the LLL cut-off bound from \cite{HKN} with the LLL fudge factor, then bound the
	% resultant of $H(x)$ and $f(x)$, and use that it must be divisible by $p^a$ because $f_i$ is a common factor
	% mod $p^a$.},
enough\footnote{If the initial value of $a$ is too low, our implementation
increases $a$, but this has little impact on the CPU time or the complexity. The highest degree term in the complexity comes from LLL
reduction. To bound the LLL cost, one must bound the vector lengths that can occur
during LLL, and the total number of LLL switches.
Gradual sublattice reduction \cite{HA} makes those
bounds independent of $a$.
%  they only depend on $n$ and the cut-off bound $n^2 \| f \|$ used in LLL\_with\_removals in \cite{HKN}.
More details can also be found in~\cite{ComplexityFactor}, which explains why the highest
degree term in the complexity of factoring in $\mathbb{Q}[x]$ depends only on $r$.

To prove an upper bound for $a$, bound the coefficients of a basis element $\beta_j \in V - L_i$ by multiplying
the LLL cut-off bound $n^2 \| f \|$ from~\cite{HKN} with the LLL fudge factor $2^{\mathcal{O}(n)}$. Then bound the
Norm of the resultant of $f(x)$ and $H(x)$ from Remark~\ref{rem:ratiorep}, and use that it must be divisible by $p^a$ because $\hat{f}_i$ is a common factor mod $p^a$ but not mod $p^{\infty}$ if $\beta_j \in V - L_i$.},
the number of CPU operations executed by Algorithm {\rm{\texttt{SubfFact}}} is bounded by \[ \tilde{\mathcal{O}}(rn^5(n+\log \|f\|)^2).\]
\end{theorem}
\begin{proof}
Steps~1 and~2 involve factoring $f$ modulo a few primes $p$ until we find a prime that satisfies the conditions from Step~1. Factoring $f$ over $\mathbb{F}_p$ can be executed with $\tilde{\mathcal{O}} (n^2+n\log p)$ operations in $\mathbb{F}_p$ (see~\cite{MCA}, Corollary~14.30). Multifactor Hensel lifting takes $\tilde{\mathcal{O}}(n^2(n + \log \|f\|_2))$ CPU operations (see~\cite{MCA}, Theorem~15.18). For each $i$ in Step~4 we have one LLL call, costing $\tilde{\mathcal{O}}(n^5(n+\log \|f\|_2)^2)$ CPU operations (see~\cite{HA}), and one \texttt{PartialSubfFact} call, which costs $\tilde{\mathcal{O}}(n^3(r+\log \|f\|_2))$ CPU operations according to Lemma~\ref{lemma:complPSF}. The theorem follows by omitting $\log p$ factors.
\end{proof}

\begin{remark}\label{rem:linear}
While computing the subfield factorization, whenever we find a linear factor $x-h_1(\alpha)\in \mathbb{Q}[\alpha][x]$ of $f$, we can use it to find new linear factors in the following way: if $x-h_2(\alpha)$ is another linear factor, then $h_1(h_2(\alpha))$ is also a root of $f$. This follows from the fact that there is a bijection between the automorphisms of $K/k$ and the roots of $f$ over $K$. This is particularly helpful when $f$ has several roots in $K$, since the number of LLL calls can be reduced significantly.
\end{remark}

\section{CPU Time Comparison}\label{sec:6}

In this last section we give a few timings comparing Algorithm \texttt{SubfFact} (Section~\ref{sec:5}) and factorization algorithms over $\mathbb{Q}(\alpha)$ (recall that both algorithms yield a subfield factorization). We also compare our algorithm \texttt{Subfields} with that from~\cite{HKN}. Our algorithm was implemented in the computer algebra system Magma, since there exists an implementation of \cite{HKN} in Magma as well.

\subsection{\texttt{SubfFact} vs. Factoring over $\mathbb{Q}(\alpha)$}\label{ssec:compare}
Algorithm \texttt{Subfields} is based on the definition of a subfield factorization of $f$. As noted before, the irreducible factorization of $f$ over $\mathbb{Q}(\alpha)$ is a subfield factorization. In this section we compare the time necessary to find a subfield factorization of $f$ using algorithm \texttt{SubfFact}, presented in Section~\ref{sec:5}, with the time necessary to completely factor $f$ over $\mathbb{Q}(\alpha)$ in Magma and in PARI/GP. We also list $s$, the number of irreducible factors of $f$ and $r$, the number of factors in the subfield factorization obtained using \texttt{SubfFact}.

\begin{table}[H]\label{table:1}
\begin{center}
\begin{tabular}{|c|c|c|r|r|r|}
\hline
 $n$ & $s$ & $r$ & \texttt{SubfFact} & \multicolumn{1}{c|}{\begin{tabular}[c]{@{}c@{}}Magma v2.21-3\\ (Factorization)\end{tabular}} & \multicolumn{1}{c|}{\begin{tabular}[c]{@{}c@{}}PARI/GP v2.9.2\\ (nffactor)\end{tabular}}  \\ \hline

 32    & 32        & 32     & 0.56s        & 4.71s &        0.46s \\ 
 36    & 24        & 16     & 3.76s        & 4.20s &        0.63s \\ 
 45    & 3         & 3      & 0.59s        & 20.01s	&	    94.54s \\
 48    & 20        & 16     & 21.10s       & 34.23s &       3.30s \\ 
 50    & 26        & 11     & 24.08s       & 20.51s &       2.89s\\
 56    & 14        & 6      & 50.26s       & 127.34s &      26.48s\\ 
 60    & 33        & 18     & 107.22s      & 1,836.80s &    38.75s \\ 
 60    & 60        & 32     & 117.43s      & 9,069.22s &   40.70s\\ 
 64    & 16        & 12     & 101.82s      & 190.99s &      48.82s\\ 
 72    &  3        & 3      & 77.76s       & 300.62s &      133.54s \\
 72    & 32        & 24     & 175.85s      & 130.40s &      17.23s \\
 75    & 20        & 6      & 542.30s      & $>24h$     &   518.40s\\ 
 75    & 21        & 9      & 199.70s      & 180.06s &      114.38s\\ 
 80    & 3         &  3     & 117.03s      & 280.18s &      136.21s \\
 81    & 42        & 28     & 680.24s      & 13,661.89s &   96.00s  \\ 
 90    & 24 	   & 7      & 921.53s      & $>24h$     &   516.14s\\
 96    & 32        & 32     & 555.24s      & 622.33s &      137.23s\\
 96    & 96        & 56     & 2,227.06s    & 16,352.01s &   91.43s\\

 \hline
\end{tabular}
\caption{Subfield Factorization vs. Factoring in $\mathbb{Q}(\alpha)[x]$.}
\end{center}
\end{table}

In a few cases, factoring $f$ over $\mathbb{Q}(\alpha)$ in Magma is faster than \texttt{SubfFact}. However, when it is not, using \texttt{SubfFact} to find a subfield factorization is usually much faster. Factoring $f$ over $\mathbb{Q}(\alpha)$ in PARI/GP is usually faster still.%, except in cases where \textit{PARI/GP} struggles to find an integral basis for $K$ (a Step that is not necessary because one can use rational univariate representation instead).
\begin{remark}
In Step~6 of Algorithm \texttt{SubfFact}, the subfield $L_i$ (to be precise: a subspace $V$ containing $L_i$, but these are practically always the same)
is computed with LLL techniques.  Factoring $f$ in PARI/GP is done with LLL techniques as well \cite{belabas}.
We expect the computation of one $L_i$ to be faster than factoring $f$ in PARI/GP, because % although the corresponding LLL reductions have similar dimension,
the LLL cut-off bound in \cite[Theorem~12]{HKN} used by \texttt{SubfFact} is practically optimal.
But the above table shows that this is not enough to  % the CPU time saved by this bound
% The above table shows that, compared with \cite{belabas}, the CPU time saved by this tight bound
compensate for the fact that Step~6 in Algorithm \texttt{SubfFact} is called $\hat{r}$ times.
In contrast, with \cite{belabas} the cost of computing one factor is the same as the cost of computing all factors. So to compute all $L_i$ it is faster to use \cite{belabas}.
%  It might als indicates what to expect if one tried to determine the complexity of \cite{belabas}.

%especially if one % replaces integral basis computations with
%switches to rational univariate representation.
	% (see \cite{HoeijSlides} for $\mathbb{Q}[x]$, but the same applies here). We could thus ask if we could outperform \cite{belabas} on a reduced problem: Assuming there is a non-trivial subfield,
	% find at least one non-trivial factor of $f$. However, it is not clear if we can outperform \cite{belabas} on this reduced question either, we would need to know which $i$ to use, otherwise we may
	% end up computing $L_i$ and finding $\mathbb{Q}$ several times before finding a non-trivial subfield.
\end{remark}

\subsection{Comparing Algorithms}\label{ssec:compare2}
In this last Section we compare the running time of algorithms \texttt{Subfields} (where the subfield factorization is computed using \texttt{SubfFact}) and the algorithm from \cite{HKN} (present in Magma). In order to give a better comparison of the running time for both algorithms, we also compute a generator for every subfield (see Section~\ref{ssec:generators}). 

As noted before, the main contribution of our work is in the way the intersections of the principal subfields are computed. In the table below $n$ is the degree of the polynomial $f$ and $\hat{r}$ is the number of irreducible factors of $f$ in $\mathbb{F}_p[x]$. We also list $r$, the number of principal subfields and $m$, the total number of subfields of the extension defined by $f$.

\begin{table}[H]\label{table:2}
\begin{center}
\begin{tabular}{|c|c|c|c|r|r|r|r|}
\hline
Ex. & $n$   & $\hat{r}$ & $r$   & $m$      & \multicolumn{1}{c|}{$m/r$} & \multicolumn{1}{c|}{{Magma v2.21-3 }} &  \multicolumn{1}{c|}{\texttt{Subfields}} \\ \hline
$f_{1}$ & 32  & 32        & 32  & 374    & 11.68       & 11.42s      & 1.15s     \\ 
$f_{2}$ &36  & 24        & 16  & 24     & 1.50        & 5.14s       & 3.84s     \\
$f_{3}$ &48  & 28        & 16  & 25     & 1.56        & 24.52s      & 21.21s    \\ 
$f_{4}$ &50  & 26        & 11  & 12     & 1.09        & 26.06s      & 24.16s    \\ 
$f_{5}$ &56  & 20        & 6   & 6      & 1.00        & 52.29s      & 50.31s    \\ 
$f_{6}$ &60  & 33        & 18  & 19     & 1.05        & 112.90s     & 107.53s   \\
$f_{7}$ &60  & 60        & 32  & 59     & 1.84        & 205.46s     & 118.50s   \\ 
$f_{8}$ &64  & 24        & 12  & 14     & 1.16        & 110.89s     & 101.99s   \\ 
$f_{9}$ &64  & 40        & 30  & 93     & 3.10        & 167.13s     & 122.24s   \\ 
$f_{10}$ &64  & 64        & 64  & 2,825  & 44.14       & 1,084.91s   & 43.62s    \\ 
$f_{11}$ &72  & 40        & 24  & 42     & 1.75        & 219.30s     & 176.65s   \\ 
$f_{12}$ &75  & 20        & 6   & 6      & 1.00        & 516.45s     & 542.60   \\ 
$f_{13}$ &75  & 21        & 9   & 10     & 1.11        & 200.42s     & 199.85s   \\ 
$f_{14}$ &80  & 48        & 27  & 57     & 2.11        & 1,021.22s   & 685.65s   \\ 
$f_{15}$ &81  & 42        & 28  & 56     & 2.00        & 715.70s     & 681.35s   \\ 
$f_{16}$ &81  & 45        & 25  & 36     & 1.44        & 746.33s     & 716.12s   \\ 
$f_{17}$ &90  & 24        & 7   & 7      & 1.00        & 923.74s     & 921.77s   \\ 
$f_{18}$ &96  & 32        & 32  & 134    & 4.18        & 1,159.04s   & 558.96s   \\ 
$f_{19}$ &96  & 96        & 56  & 208    & 3.71        & 4,026.65s   & 2,239.54s \\ 
$f_{20}$ &100 & 100       & 57  & 100    & 1.75        & 7,902.09s   & 4,250.39s \\ 
$f_{21}$ &128 & 128       & 128 & 29,211 & 228.21      & 306,591.68s & 5,164.75s \\ \hline
\end{tabular}

   \caption{Comparison table.}
   \end{center}
 \end{table}

Notice that when $m$ is close to $r$ (i.e., when there are not many subfields other than the principal subfields and hence, very few intersections to be computed) our algorithm performs similarly as~\cite{HKN}. However, we see a noticeable improvement when $m$ is very large compared to $r$, since in this case there are a large number of intersections being computed. 

We remark that the time improvements are not only due to the new intersection method, but also due to Remark~\ref{rem:linear}. In some cases, several LLL calls can be skipped, which greatly improves CPU times (this is the case for the Swinnerton-Dyer polynomials $f_1,f_{10}$ and $f_{21}$, where the number of LLL calls was $5,6$ and $7$, respectively). The implementation of our algorithm, as well as the polynomials used in this comparison table, can be found in~\cite{impl}.

\appendix 

\section{Subfield Polynomial Equivalences}\label{App:A}

% Section~1.1 we defined
The {\em subfield polynomial} of a subfield $L$ of $K/k$ is the polynomial that satisfies any of the 6 properties listed in Theorem~\ref{theo:equiv}.
This appendix shows that these properties are equivalent.

\begin{lemma}\label{lemma:degree}
If $g\in L[x]$ is the minimal polynomial of $\alpha$ over $L$, for some subfield $L$ of $K/k$, then
\begin{enumerate}
\item[$(i)$] $\deg(g) \cdot [L : k] =n.$
\item[$(ii)$] $L=\{h(\alpha)\in K : h(x) \equiv h(\alpha)\bmod g \} = L_g.$
\item[$(iii)$] The coefficients of $g$ generate $L$ over $k$.
\end{enumerate}
\end{lemma}
\begin{proof}
Item $(i)$ is the \emph{product formula} $[K:L] \cdot [L:k] = [K:k]$. For a proof of $(ii)$ see \cite{HKN}, Theorem~1. For $(iii)$,
let $\tilde{L} \subseteq L$ to be the field generated over $k$ by the coefficients of $g$. Then $g$ is also the minimal polynomial of $\alpha$ over $\tilde{L}$,
and hence $[K:\tilde{L}] = \deg(g) = [K:L]$ so $[L:\tilde{L}]=1$ and hence $L = \tilde{L}$.
\end{proof}

% MOVED:
% \begin{lemma}\label{remark:coeff}
% If $g\in L[x]$ is the minimal polynomial of $\alpha$ over $L$ then the coefficients of $g$ generate $L$ over $k$.
% \end{lemma}
% \begin{proof}
% The minimal polynomial of $\alpha$ over $L$ has degree $[L(\alpha):L] = [K:L]$.
%
% Since $\tilde{L} \subseteq L$, the minimal polynomial of $\alpha$ over $\tilde{L}$ has degree $[K:\tilde{L}] \geq [K:L] = \deg(g)$, and hence equals $g$. Therefore $[K:\tilde{L}] = [K:L]$ and hence $L = \tilde{L}$.
% \end{proof}

\begin{lemma}\label{lemma:gcd}Let $h(x) \in k[x]$ and let $L = k( h(\alpha) )$ be a subfield of $K/k$. The minimal polynomial of $\alpha$ over $L$ is the gcd of $f$ and $h(x) - h(\alpha)$.
\end{lemma}
\begin{proof}
Let $g$ be the gcd, $d$ its degree, and let $\tilde{g}$ be the minimal polynomial of $\alpha$ over $L$.
The polynomials $f$, $h(x)-h(\alpha)$, and $g$, are elements of $L[x]$ and have $\alpha$ as a root, and are thus divisible by $\tilde{g}$.
It remains to show that $g$ and $\tilde{g}$ have the same degree. If $\alpha_1,\ldots,\alpha_n$ are the roots of $f$ in a splitting field, then the
roots of $g$ are those $\alpha_i$ for which $h(\alpha_i) = h(\alpha)$. So $d$ is the number of $i$ for which $h(\alpha_i) = h(\alpha)$.
The degree $[L : k]$ is the number of distinct $h(\alpha_i)$, which is $n/d$. The degree of $\tilde{g}$ is $[K:L] = n/[L:k] = d$.
\end{proof}

\begin{lemma}\label{lemma:apx}
Let $L_g$ be a subfield of $K/k$, for some $g\mid f$ and let $\tilde{g}$ be the minimal polynomial of $\alpha$ over $L_g$. Then $g\mid \tilde{g}$.
\end{lemma}
\begin{proof}
Let $h\in k[x]$ be such that $L_g = k(h(\alpha))$. By the previous lemma, $\tilde{g} = {\rm gcd}(f, h(x)-h(\alpha))$ which is divisible by $g$.
% 
% =\prod_{i\in I} f_i$, where $I=\{ 1\leq i\leq r \; : \; h(x)\equiv h(\alpha) \bmod f_i \}$
% % (See~\cite{HKN}, Lemma~14)
% , where $f_1,\ldots, f_r$ is a factorization of $f$ into irreducible factors over $K$. On the other hand, $g=\prod_{i\in T} f_i$, for some $T\subseteq \{1,\ldots, r\}$, and $h(x)\equiv h(\alpha) \bmod g$. By separability and the Chinese remainder theorem, it follows that $T\subseteq I$ and hence, $g\mid \tilde{g}$.
\end{proof}

\noindent \textbf{Proof.} (of Theorem~\ref{theo:equiv}, Section~\ref{sec:2}). Lemma~\ref{lemma:gcd} shows $1)\Leftrightarrow 6)$.
We shall prove that $1)\Rightarrow 2) \Rightarrow 3) \Rightarrow 4) \Rightarrow 5) \Rightarrow 1)$.
\begin{description}
\item[$1) \Rightarrow 2)$] Follows from Lemma~\ref{lemma:degree}.

\item[$2) \Rightarrow 3)$] Let $\tilde{g}$ be the minimal polynomial of $\alpha$ over $L:=k(\text{coeffs}(g))$. Thus, $L=L_{\tilde{g}}$. Moreover, since $g,\tilde{g}\in L[x]$ and $g(\alpha)=0$, we have $\tilde{g} \mid g$. Hence, \[ n=\deg(\tilde{g})\cdot [L_{\tilde{g}}:k]=\deg(\tilde{g})\cdot [L:k] \leq \deg(g)\cdot [L:k] \leq n. \] Thus, $g=\tilde{g}$. Item $3)$ then follows from Lemma~\ref{lemma:degree}~$(i)$.

\item[$3) \Rightarrow 4)$] Let $\tilde{g}$ be the minimal polynomial of $\alpha$ over $L_g$. Thus, $L_{\tilde{g}} = L_g$. By Lemma~\ref{lemma:apx}, we have $g\mid \tilde{g}$. Hence, \[n=\deg(g)\cdot [L_{g} : k] =\deg(g)\cdot [L_{\tilde{g}}:k] \leq \deg(\tilde{g})\cdot [L_{\tilde{g}} : k]=n.\] Thus, $g=\tilde{g}$. Item $4)$ then follows from Lemma~\ref{lemma:degree}~$(iii)$.

\item[$4) \Rightarrow 5)$] Trivial.

\item[$5) \Rightarrow 1)$] Let $\tilde{g}$ be the minimal polynomial of $\alpha$ over $L_g$. By Lemma~\ref{lemma:apx} it follows that $g\mid \tilde{g}$. On the other hand, since $g\in L_g[x]$ and $g(\alpha)=0$, we have $\tilde{g}\mid g.$ Therefore, $g=\tilde{g}$ and item $1)$ follows. 
% \item[$1)\Leftrightarrow 6)$] Finally, if $1)$ holds, let $\sum c_i \alpha^i$ be a primitive element of $L_g$, for some  $c_i \in k$. Let $h(x) = \sum c_i x^i$ and define $H=h(x)-h(\alpha).$ Hence $L_g = k(h(\alpha))$ and since $H\in L_g[x]$ and $H(\alpha)=0$, it follows that $g\mid \gcd(f,H)$. Let $f_1, \ldots, f_r$ be the factorization of $f$ into irreducible factors over $K$. Suppose that $f_j\mid \gcd(f,H)$ but $f_j \nmid g$. Since $f_j \mid h(x)-h(\alpha)$, it follows that $h(\alpha)\in L_{f_j}$ and hence, $L_g \subseteq L_{f_j}$. Let $F_j$ be the minimal polynomial of $\alpha$ over $L_{f_j}$. By Lemma~\ref{lemma:apx}, $f_j \mid F_j \mid g$, which is a contradiction. Conversely, suppose that $6)$ holds. Since $g\mid h(x)-h(\alpha)$, it follows that $h(\alpha)\in L_g$. Hence $h(x)-h(\alpha)\in L_g[x]$ and therefore, $g\in L_g[x]$. Hence, $5)$ holds (which we already proved to be equivalent to $1)$.
\qed\end{description}

\section{Primitive Element Probability}\label{sec:prob}

Let $L/k$ be a separable field extension and let $\beta_1,\ldots, \beta_m$ be a $k$-basis of $L$. Let $T\subseteq k$ finite and let $S=\{\sum a_i \beta_i : a_i\in T\}$. In this section we compute the probability that a random element $s\in S$ is a primitive element of $L$.

\begin{lemma}\label{lemma:p1}
Let $V$ be a $k$-vector space with basis $v_1,\ldots, v_m$. Let $W\subseteq V$ be a subspace of dimension $d$. Let $T\subseteq k$ be a finite set and let $S=\{\sum_{i=1}^m a_i v_i : a_i\in T\}.$ Then \[ |S\cap W| \leq |T|^d. \]
\end{lemma}
\begin{proof}
Let $w_1,\ldots, w_d$ be a basis of $W$. For every $j$ there exist $c_{i,j}\in k$, $1\leq i \leq m$, such that \begin{equation}\label{eq1teop1}
w_j=\sum_{i=1}^m c_{i,j} v_i.
\end{equation} Let $w\in W$, then \begin{equation}\label{eq2teop1}
 w=\sum_{i=1}^m a_i v_i = \sum_{j=1}^d b_j w_j,
 \end{equation} for some $a_i\in k$, $1\leq i \leq m$ and some $b_j\in k$, $1\leq j \leq d$. Combining equations (\ref{eq1teop1}) and (\ref{eq2teop1}), it follows that $ a_i = \sum_{j=1}^d c_{i,j} b_j , \; 1\leq i \leq n. $ That is, we have the following equation \begin{equation}\label{eq3teop1} 
 \left( \begin{array}{c}
a_1 \\
\\
\vdots \\
\\
a_m
\end{array} \right)= \left( \begin{array}{ccc}
 c_{1,1} & \cdots & c_{1,d} \\
  & & \\
 \vdots & & \vdots \\
 & & \\
 c_{m,1} & \cdots & c_{m,d}\end{array}  \right)  \left( \begin{array}{c}
b_1 \\
\vdots \\
b_d 
\end{array}  \right). \end{equation} If $C$ is the $m\times d$ matrix in~(\ref{eq3teop1}), then $C$ has $d$ linearly independent rows. That is, only $d$ of the values $a_i$ suffice to determine $w$, while the remaining values are dependent. Therefore, \[ |S\cap W | \leq |T|^d. \]
\end{proof}

\begin{theorem}\label{teo:c1}
Let $L/k$ be a separable field extension and let $\beta_1,\ldots, \beta_m$ be a $k$-basis of $L$. If $T\subseteq k$ is a finite set and $S=\{\sum a_i \beta_i : a_i\in T\}$, then \[ |\{s\in S : k(s)\subsetneq L\}| \leq (m-1)\cdot |T|^{m/p},\]
where $p$ is the smallest prime that divides $m$.
\end{theorem}

\begin{proof}
Let $L_1,\ldots, L_r$ be the principal subfields of $L/k$. Since every subfield of $L/k$ is an intersection of some of the principal subfields of $L/k$, it suffices to find $|\{s\in S : s \in L_i\subsetneq L,\text{ for some } 1\leq i \leq r \}|$.

The number of principal subfields (not equal to $L$) is at most $m-1$ and $[L_i : k]\leq m/p$, where $p$ is the smallest prime that divides $m$. According to Lemma~\ref{lemma:p1}, $|S\cap L_i| \leq |T|^{m/p}$. Therefore, $ |\{s\in S : k(s)\subsetneq L\}|\leq (m-1)\cdot |T|^{m/p}$.
\end{proof}

\begin{corollary}\label{cor:prob}
Let $L/k$ be a separable field extension and let $\beta_1,\ldots, \beta_m$ be a $k$-basis of $L$. Let $T\subseteq k$ finite and let $S=\{\sum a_i \beta_i : a_i\in T\}$. If $s$ is a random element of $S$ and $p$ is the smallest prime that divides $m$, then \[\text{Prob}(k(s)\subsetneq L)\leq (m-1)\cdot |T|^{m(1-p)/p}.\]
\end{corollary}

\section{GCD's in $\mathbb{Q}(\alpha)[x]$}\label{App:C}

One ingredient of Algorithm~\texttt{PartialSubfFact} in Section~\ref{sec:5} is
computating $\gcd$'s in $\mathbb{Q}(\alpha)[x]$. If $p$ is a prime, $\mathbb{F}_p(\alpha):=\mathbb{F}_p[t]/(f(t))$ is a finite ring. Let $g_1, g_2 \in \mathbb{Q}(\alpha)[x]$. The modular $\gcd$ algorithm reconstructs $g:=\gcd(g_1, g_2)$ from its images in $\mathbb{F}_p(\alpha)[x]$ for suitable primes. In other words, there are mainly four steps to be carried out (see~\cite{modular})
\begin{itemize}
\item[1)] Compute $g_1\bmod p$, $g_2 \bmod p$, for several suitable primes $p$.
\item[2)] Compute $\gcd(g_1\bmod p, g_2 \bmod p)$, for each prime $p$.
\item[3)] Chinese remainder the polynomials in $2)$ and use rational reconstruction to find a polynomial $g\in \mathbb{Q}(\alpha)[x].$
\item[4)] Trial Division: check if $g|g_1$ and $g|g_2$.
\end{itemize}

The number of primes needed depends on the coefficient size of $g$. But the (bound for) coefficient size of $f'(\alpha)g\in \mathbb{Z}[\alpha][x]$ is much better than that of $g\in \mathbb{Q}(\alpha)[x]$. Hence, to get a good complexity/run time we choose to reconstruct $f'(\alpha) g$ from its modular images instead of $g$.
Furthermore, if we have some information about $g$ (for instance, its degree), then Step $4$ can be skipped.

We need to compute $\gcd(H, g_j)$, where $H=f'(x)\tilde{h}(\alpha)-f'(\alpha)\tilde{h}(x)\in \mathbb{Z}[\alpha][x]$ is as in Remark~\ref{rem:ratiorep} and $g_j$ is a factor of $f$ over $\mathbb{Q}(\alpha)$. For $\beta\in \mathbb{Q}(\alpha)$ let \[T_2(\beta):=\sum_{i=1}^n |\sigma_i(\beta)|^2\] be the $T_2$-norm of $\beta$, where $\sigma_i$ is the $i$-th embedding of $\mathbb{Q}(\alpha)$ into $\mathbb{C}$, and if $f'(\alpha)\beta = \sum b_i \alpha^i\in \mathbb{Z}[\alpha]$, define  \[ \|f'(\alpha)\beta\|_2=\|\sum b_i\alpha^i\|_2:=\|(b_0, \ldots, b_{n-1})\|_2.\] These two norms can be related in the following way (see \cite{HKN}, Lemma~18)
\begin{equation}\label{eq:comparenorms}
\|(f' (\alpha)\beta\|_2 \leq n^{3/2} \|f\|_2 \sqrt{T_2(\beta)}.
\end{equation} 

Let us first bound the integer coefficients of $f'(\alpha)g$ and $f'(\alpha)\text{lcoeff}(H)G$, where $\text{lcoeff}(H)$ is the leading coefficient of $H$, $g$ is a monic factor of $f$ and $G$ is a monic factor of $H$. 

\begin{lemma}\label{lemma:bound2}
Let $g\in \mathbb{Q}(\alpha)[x]$ be a factor of $f$ and let $c$ be a coefficient of $g$. Then  \[ \| f'(\alpha)c\|_2 \leq n4^n\|f\|_2^2. \]
\end{lemma}
\begin{proof}
See~\cite{isomorphisms}, Section~4.
\end{proof}

\begin{lemma}\label{lemma:bound3}
Let $G\in \mathbb{Q}(\alpha)[x]$ be a monic factor of $H$ and let $c$ be a coefficient of $G$. Then $f'(\alpha)\text{lcoeff}(H)c \in \mathbb{Z}[\alpha]$ and \[ \|f'(\alpha)\text{lcoeff}(H)c \|_2 \leq  n^{7.5} T_B2^{n+1}\|f\|_2^3 (1+\|f\|_2)^n, \] where $T_B$ is a bound for the elements in $T$ (with $H$ as in Remark~\ref{rem:ratiorep}).
\end{lemma}
\begin{proof}
This bound follows from Equation~(\ref{eq:comparenorms}) and the Mignotte bound for coefficients of complex factors of $\sigma_i(H)$.
\end{proof}

 Let us now determine the cost (in CPU operations) for computing \[f'(\alpha)\gcd(H,f'(\alpha)g_j)\in \mathbb{Z}[\alpha][x].\] Let $B$ bound the integer coefficients of $f'(\alpha)\gcd(H,f'(\alpha)g_j)$ (Lemma~\ref{lemma:bound2}).

\begin{itemize}
\item[1)]
 First we have to compute the images of $H$ and $f'(\alpha)g_j$ in $\mathbb{F}_p(\alpha)[x]$, which can be done with $\mathcal{O}(n^2)$ integer reductions modulo several primes $p$. The number of primes is $\mathcal{O}(\log B)=\tilde{\mathcal{O}}(n+\log\|f\|_2)$. According to \cite{MCA}, Theorem~10.24, the complexity of this step can be bounded by \[ \tilde{\mathcal{O}}(n^2(n+\log \|f\|_2)). \]

\item[2)]
  Next we have to compute one $\gcd$ in $\mathbb{F}_p(\alpha)[x]$, for $\mathcal{O}(\log B)$ primes $p$. Using the Extended Euclidean Algorithm (see \cite{MCA}, Corollary 11.6), one $\gcd$ in $\mathbb{F}_p(\alpha)[x]$ can be computed with $\tilde{\mathcal{O}}(n)$ operations in $\mathbb{F}_p(\alpha)$ or $\tilde{\mathcal{O}}(n^2)$ operations in $\mathbb{F}_p$. Hence, the complexity of this step can be bounded by \[ \tilde{\mathcal{O}}(n^2(n+\log \|f\|_2)). \]
 
\item[3)]
In this step we need to find a polynomial $f'(\alpha)G\in \mathbb{Z}[\alpha][x]$ whose images modulo several primes are given in Step 2. For this we use the Chinese Remainder Algorithm (CRA). There are $n(d+1)$ integers to be reconstructed, where $d=\deg(\gcd(H, g_j))$, and each CRA call costs $\tilde{\mathcal{O}}(\log P)$, where $P=\prod p$ (see \cite{MCA}, Theorem~10.25). Since $P=\mathcal{O}(B)$, the total cost of this step is \[ \tilde{\mathcal{O}}(n^2(n+\log \|f\|_2)). \]

\item[4)]
Instead of computing the division $H/G$ (and $g_j/G$, whose complexity is hard to bound), we can substitute this trial division by reconstruction from modular images followed by a trial multiplication. That is, we can compute the images of $H$ and $G$ modulo several primes $p$, compute $H/G$ modulo $p$ and then reconstruct $f'(\alpha)\text{lcoeff}(H)(\frac{H}{G})\in \mathbb{Z}[\alpha][x]$ and verify that $f'(\alpha)\text{lcoeff}(H)(\frac{H}{G})\cdot G = f'(\alpha)H$. The cost is similar to steps 1), 2) and 3) above, the only difference is the number of primes needed (since what we want to reconstruct is a factor of $H$, the bound $B$ is given by Lemma~\ref{lemma:bound3}) and the trial multiplication at the end (which can be executed with $\tilde{\mathcal{O}}(n^3\log \|f\|_2)$ CPU operations). Hence, this step has complexity \[ \tilde{\mathcal{O}}(n^3\log \|f\|_2). \]
\end{itemize}

\end{document}